\documentclass[showpacs,preprintnumbers,amsmath,amssymb,
floatfix,nofootinbib]{revtex4-1}

\usepackage{epsfig}
\usepackage{amsmath,amssymb,amsfonts,amsthm}
\usepackage{mathrsfs}
\usepackage{youngtab}
\usepackage[hang,nooneline]{subfigure}
\usepackage{pstricks,pstricks-add,pst-plot}

\numberwithin{equation}{section}
\newtheorem{theorem}{Theorem}[section]
\newtheorem{lemma}[theorem]{Lemma}
\newtheorem{corollary}[theorem]{Corollary}
\newtheorem{proposition}[theorem]{Proposition}

\allowdisplaybreaks \numberwithin{equation}{section}

\newcommand{\vek}[1]{\boldsymbol{#1}}

\newcommand{\weglassen}[1]{}

\renewcommand{\imath}{\mathrm{i}}
\renewcommand{\d}{\mathrm{d}}

\begin{document}

\title[Hyperelliptic integrals and Ho\v{r}ava-Lifshitz black hole space-times]
{Inversion of a general hyperelliptic integral and particle motion
in Ho\v{r}ava-Lifshitz black hole space-times}

\author{Victor Enolski$^{1,2,5}$, Betti Hartmann$^3$, Valeria Kagramanova$^4$, Jutta Kunz$^4$, Claus L\"ammerzahl$^{5,4}$, Parinya Sirimachan$^3$}
\affiliation{
$^1$ Hanse-Wissenschaftskolleg (HWK), 27733 Delmenhorst, Germany \\
$^2$ Institute of Magnetism, 36-b Vernadsky Blvd, Kyiv 03142, Ukraine \\
$^3$ School of Engineering and Science, Jacobs University Bremen, 28759 Bremen, Germany \\
$^4$ Institut f\"ur Physik, Universit\"at Oldenburg,
D--26111 Oldenburg, Germany
\\
$^5$ ZARM, Universit\"at Bremen, Am Fallturm,
D--28359 Bremen, Germany}

\email{V.Z.Enolskii@ma.hw.ac.uk}

\email{b.hartmann@jacobs-university.de}

\email{va.kagramanova@uni-oldenburg.de} %

\email{jutta.kunz@uni-oldenburg.de} %

\email{laemmerzahl@zarm.uni-bremen.de}

\email{p.sirimachan@jacobs-university.de}

\date\today

\begin{abstract}
The description of many dynamical problems like the particle motion in higher dimensional spherically and axially symmetric space-times is reduced to the inversion of hyperelliptic integrals of all three kinds. The result of the inversion is defined locally, using the algebro-geometric techniques of the standard Jacobi inversion problem and the foregoing restriction to the $\theta$-divisor. For a representation of the hyperelliptic functions the Klein--Weierstra{\ss} multivariable $\sigma$-function is introduced. It is shown that all parameters needed for the calculations like period matrices and abelian images of branch points can be expressed in terms of the periods of holomorphic differentials and $\theta$-constants. The cases of genus two, three and four are considered in detail. The method is exemplified by the particle motion associated with genus one elliptic and genus three hyperelliptic curves. Applications are for instance solutions to the geodesic equations in the space-times of static, spherically symmetric Ho\v{r}ava-Lifshitz black holes.

\end{abstract}

\maketitle


\section{Introduction}\label{sec:introduction}

\subsection{The mathematical problem}

Various problems of physics are reduced to the inversion of a hyperelliptic integral. Namely, let
\begin{equation}
y^2=4x^{2g+1}+\lambda_{2g}x^{2g}+\ldots+\lambda_0 \label{curve}
\end{equation}
be a hyperelliptic curve $X_g$ of genus $g$ with one branch point at infinity realized as a two-sheeted covering over the extended complex plane. A point $P \in X_g$ has coordinates $P=(x,y)$, where the sign of the second coordinate $y$ indicates the chosen sheet on a Riemann surface. Let $\mathcal{R}(x,y)$ be a rational function of its arguments $x$ and $y$. We consider here the problem of the inversion of an abelian integral 
\begin{equation}
\int_{x_0}^x \mathcal{R}(x,y) \mathrm{d}x=t \label{inversion1} 
\end{equation}
resulting in a function $x(t)$ which is a function of the complex variable $t$.
The integral~\eqref{inversion1} can be decomposed by routine algebraic operations to \begin{equation}
\mathcal{E}(x)-\mathcal{E}(x_0)+
\sum_{k=1}^g a_k \int_{x_0}^x \mathrm{d} u_k + \sum_{k=1}^g b_k \int_{x_0}^x \mathrm{d} r_k
+\sum_{k=1}^n c_k \int_{x_0}^x \mathrm{d} \Omega_{\alpha_k,\beta_k} = t \ , \label{inversion2}
\end{equation}
where $\mathcal{E}(x)$ is an elementary function including logarithms and rational functions,
$a_k$, $b_k$, $c_k$ are certain constants and $\mathrm{d}u_k,\mathrm{d}r_k$ and  $\mathrm{d}\Omega_{\alpha_k,\beta_k}$ are differentials of the first, second and third kind, respectively. Namely, $\mathrm{d}u_k$ are holomorphic differentials,  $\mathrm{d}r_k$ are meromorphic differentials of the second kind with a unique pole of the order $2g-2k+2$, and $\mathrm{d}\Omega_{\alpha_k,\beta_k}$ are meromorphic differentials of the third kind with first order poles in the points $\alpha_k$ and $\beta_k$ and residues $\pm1$ in the poles. We suppose that in the case considered there are $n\geq 0$ differentials of the third kind.

It is well known that only in the case of elliptic curves, i.e. for $g=1$, the correspondence $x \leftrightarrow t$ is one-to-one and the aforementioned inversion problem can be
solved in terms of single-valued elliptic functions. In the case of
higher genera, $g>1$, a one-to-one correspondence is achieved between the symmetrized products
of curves $X_g\times\ldots\times X_g$ and a multi-dimensional complex space, the Jacobi variety. Single-valued functions in this case appear to be multi-periodic functions of many complex variables, called abelian functions. These ideas already developed by Jacobi had led Riemann
to the concept of the Riemann surfaces, to the introduction of multi-dimensional $\theta$-functions, and the formulation of his celebrated theorems.

Moreover it is also known that the function $x(t)$ becomes single-valued on the infinitely-sheeted Riemann surface. In particular, when abelian integrals reduce to elliptic integrals, this Riemann surface becomes finitely-sheeted~\cite{GS07}. For the special case of a genus two hyperelliptic curve a detailed construction of such an infitely-sheeted Riemann surface surface was given by Fedorov and G\'omes-Ulate~\cite{fg07}. It also follows from~\cite{fg07} that $x(t)$ is well defined on the complex plane from which an infinite lattice of polygons with $(4g-4)$-edges called ``windows'' is extracted. In our investigation we are considering a function $x(t)$ defined on the complement to these windows and suppose that the integration paths in~\eqref{inversion2} never intersect these windows.

Our approach to the inversion of hyperelliptic integrals is based on the well developed theory of hyperelliptic abelian functions and on various relations between the $\theta$-functions and $\theta$-constants. We are describing the function $x(t)$ as the restriction of an abelian function, which can be expressed in terms of symmetric functions of the divisor in the associated Jacobi inversion problem, to the one-dimensional stratum of the $\theta$-divisor. We are implementing the Klein-Weierstra{\ss} realization of hyperelliptic functions in terms of multi-variative $\sigma$-functions that represent a natural generalization of the standard Weierstra{\ss} $\sigma$-function to hyperelliptic curves of higher genera.

This paper continues our recent work~\cite{ehkkl11} where the inversion of hyperelliptic holomorphic integrals was considered. The novelty of our approach is the simultaneous consideration of all three kinds of abelian integrals from a unified viewpoint. At the heart of our method lie algebraic expressions of the symmetric bi-differential of the second kind involving the explicitly given Kleinian 2-polar. The inversion procedure involves the theta-constant relations (the Thomae and Bolza formulae) that manifest the link between branch points and Riemann period matrices. We are also developing a computer algebra procedure that allows to use Maple/algcurves software without explicit knowledge of a homology basis encrypted in the Tretkoff-Tretkoff algorithm. This will enable us to find all needed quantities like the vector of Riemann constants, the period matrices of the first and second kind, as well as the correspondence between branch points and $\theta$-characteristics. We emphasize that in this investigation we concentrate on the algebraic side of the derivation The function $x(t)$ obtained in that way is defined only locally while its analytic continuation never intersects the infinite set of cuts introduced in~\cite{fg07}. 

Another approach to the problem of inversion of integrals of the second and third kind that is based on the generalized $\theta$-function goes back to Clebsch and Gordan~\cite{cg866} and was developed in~\cite{fed99, bf08}. Here we do not discuss generalized Jacobians, and we plan to make a comparison between these two methods in another publication.

\subsection{Physical motivation}

The mathematical results described in this paper have direct applications to the solution
of the geodesic equation in certain Ho\v{r}ava-Lifshitz black hole space-times. The Ho\v{r}ava-Lifshitz theory~\cite{horava1,horava2} is an alternative
gravity theory that is powercountable renormalizable. The basic idea is that only higher spatial derivative terms are added, while higher temporal derivatives which would lead to ghosts are not considered. This leads unavoidably to the breaking of Lorentz invariance at short distances. Static and spherically symmetric black hole solutions have been studied in this theory~\cite{ks,lu_mei_pope,park}. Considering the Ho\v{r}ava--Lifshitz theory as a modification of General Relativity, one can study the solutions of the geodesic equation in the Ho\v{r}ava-Lifshitz black hole space-times. In this paper, we are mainly interested in one of
the black hole space--times given in \cite{lu_mei_pope}.

The mathematical techniques described in this paper can not only be used to solve analytically the geodesic equation in the Ho\v{r}ava--Lifshitz space--time considered here. The differentials of the first and third kind with underlying polynomial curves of arbitrary genus appear in the geodesic equations in many general relativistic space--times. The holomorphic differentials appear in the equations for the $r$- and $\vartheta$-coordinates, while the differentials of the third kind appear in the equations for the $\varphi$- and $t$-coordinates. This is also the case e.g.~in the geodesic equations for neutral particles in Taub-NUT~\cite{KKHL10} space-times and in the space-times of Schwarzschild and Kerr black holes pierced by a cosmic string~\cite{Hackmannetal10}, as well as for charged particles in the Reissner-Nordstr\"om~\cite{GK10} space-time, where elliptic integrals of the first and third kind appear. They also appear in the Schwarzschild--de Sitter \cite{HackmannLaemmerzahl08a,HackmannLaemmerzahl08} and Kerr--de Sitter  space-times~\cite{Hackmannetal010} as well as in generalized black hole Pleba\'nski--Demia\'nski space-times in 4 dimensions~\cite{Hackmannetal09} with underlying hyperelliptic curves of genus two in the geodesic equations. Also in the higher dimensional space--times of Schwarzschild, Schwarzschild--de Sitter, Reissner-Nordstr\"om and Reissner-Nordstr\"om-de Sitter~\cite{Hackmannetal08,ehkkl11} this powerful mathematics of the theory of hyperelliptic functions of higher genera is successfully applicable. Geodesics in  higher dimensional axially symmetric space-times, the Myers-Perry space-times, are integrated by the hyperelliptic functions of arbitrary genus as well~\cite{ehkkl11}. In~\cite{ehkkl11} the integration of holomorphic integrals for any genus of the underlying hyperelliptic polynomial curve has been presented. Here we expand our considerations and present the solution for the integrals of the third kind for arbitrary genus as well.

\subsection{Outline of the paper}

The paper is organized as follows. Section~\ref{sec:hyfu} represents a short introduction to the
theory of hyperelliptic functions that is adjusted to the aim of the paper. In this section we develop the $\sigma$-functional realizations of hyperelliptic functions. The key-formula in Lemma~\ref{lemma:zeta-formula} relates the integral of the second kind to the $\zeta$-function and the $\mathfrak{Z}$-vector. In Section~\ref{sec:invhige} we consider the inversion of the holomorphic integral, integrals of the second kind and third kind, and also their arbitrary combination. Section~\ref{sec:coalsume} shows that already developed means of computer algebra, like Maple/algcurves, are sufficient to compute all period matrices relevant to the $\sigma$-functional approach. As a particular feature the explicit knowledge of the homology basis ciphered in the software is not necessary.  Sections~\ref{sec:curvesgenus2}-\ref{sec:curvesgenus4} exemplify the developed method in the case of genus two, three and four, correspondingly. Section~\ref{sec:applHL} is devoted to the application of the developed method to the problem of geodesic motion in Ho\v{r}ava-Lifshitz black hole space-times. We conclude in Section~\ref{sec:conclusions}.

\section{Hyperelliptic abelian functions}\label{sec:hyfu}

\subsection{Abelian differentials and their periods}

We first introduce a canonical homology basis of cycles $(\mathfrak{a}_1,\ldots, \mathfrak{a}_g;\mathfrak{b}_1,\ldots, \mathfrak{b}_g)$, $\mathfrak{a}_i\cap \mathfrak{a}_j=\mathfrak{b}_i\cap \mathfrak{b}_j =\emptyset $, $\mathfrak{a}_i\cap \mathfrak{b}_j=-\mathfrak{b}_i\cap \mathfrak{a}_j =\delta_{ij}$, where $\delta_{ij}$ is the Kronecker symbol and $\cap$ denotes the intersection of cycles. We denote by $\mathrm{d}\boldsymbol{u}(P)=(\mathrm{d}u_1(P),\ldots,\mathrm{d}u_g(P))^T$ the basis set of holomorphic differentials (of the first kind)
\begin{equation} \mathrm{d}u_i=\frac{x^{i-1}}{y}\mathrm{d}x,\quad j=1,\ldots,g \, ,  \label{holomorphicdiff}
\end{equation}
and by $\mathrm{d}\boldsymbol{r}(P)=(\mathrm{d}r_1(P),\ldots,\mathrm{d}r_g(P))^T$ the associated meromorphic differentials (of the second kind) with a unique pole at infinity,
\begin{equation} \mathrm{d}r_i=\sum_{k=i}^{2g+1-i}(k+1-i)\lambda_{k+1+i}\frac{x^{k}}{4y}\mathrm{d}x,\quad i=1,\ldots,g \, ,\label{meromorphicdiff}
\end{equation}
where the coefficients $\lambda_i$ are as in Eq.~\eqref{curve}. These canonical holomorphic differentials $\mathrm{d} \boldsymbol{u}$ and associated meromorphic differentials $\mathrm{d} \boldsymbol{r}$ are chosen in such a way that their $g \times g$ period matrices in the fixed homology basis
\begin{align} \label{periods}\begin{split}
 2\omega  &= \Bigl(\oint_{{\mathfrak a}_k} \mathrm{d} u_i\Bigr)_{i,k=1,\ldots,g}, \qquad
 2\omega' = \Bigl(\oint_{{\mathfrak b}_k} \mathrm{d} u_i\Bigr)_{i,k=1,\ldots,g} \\
 2\eta    &= \Bigl(-\oint_{{\mathfrak a}_k}\mathrm{d} r_i\Bigr)_{i,k=1,\ldots,g} ,\qquad
 2\eta'   = \Bigl(-\oint_{{\mathfrak b}_k}\mathrm{d} r_i\Bigr)_{i,k=1,\ldots,g}\end{split}
\end{align}
satisfy the generalized Legendre relation
\begin{equation}
MJM^T=-\frac{\imath\pi}{2} J
 \label{Legendre}
\end{equation}
with
\begin{equation} \label{defM}
M = \begin{pmatrix}
\omega & \omega' \\ \eta & \eta' \end{pmatrix} \, , \qquad
J = \begin{pmatrix} 0_g & -1_g \\ 1_g  & 0_g \end{pmatrix} \, ,
\end{equation}
where $0_g$ and $1_g$ are the zero and unit $g \times g$--matrices. 

The differential of the third kind  $\Omega_{P_1,P_2}(P)$ with poles at {\em finite} points $P_1=(a_1,y_1)$ and $P_2=(a_2,y_2)$ and residues $+1$ and $-1$, respectively, can be given in the form\footnote{One can add an arbitrary combination of holomorphic differentials. However, we take this form as the most simple one which is sufficient for the following derivations.}
\begin{equation}
\Omega_{P_1,P_2}(P)=\frac{y+y_1}{2(x-a_1)}\frac{\mathrm{d}x}{y}
-\frac{y+y_2}{2(x-a_2)}\frac{\mathrm{d}x}{y}.\label{thirdkind}
\end{equation}
If the poles $P_1,P_2$ have the same $x$--coordinate and lie on different sheets, i.e.
$P_{1}=(a,y(a))$ and $P_{2}=(a,-y(a))$, then~\eqref{thirdkind} takes the form
\begin{equation}
\Omega_{P_1,P_2}(P)=\frac{y(a)}{x-a}\frac{\mathrm{d}x}{y}.\label{thirdkind1}
\end{equation}
The differentials $\mathrm{d}u_k$, $\mathrm{d}r_k$  and $\Omega_{P_1,P_2}(P)$ given above describe the entries in the relation (\ref{inversion2}).

We introduce the fundamental bi-differential $\Omega(Q,S)$ on $X_g\times X_g$ which is uniquely defined by the following conditions:
\begin{enumerate}
\item  It is symmetric, $\Omega(Q,S)=\Omega(S,Q)$. 
\item It has the poles along the diagonal $Q=S$, namely, if $\xi(Q)$ and $\xi(S)$ are local coordinates of the points $Q$ and $S$ in the vicinity of the point $P$ ($\xi(P)=0$) then the following expansion is valid
\begin{equation}
\Omega(Q,S) = \frac{\mathrm{d}\xi(Q)\mathrm{d}\xi(S)}{(\xi(Q)-\xi(S))^2} + \sum_{m,n\geq 1} \Omega_{mn}(P) \xi(Q)^{m-1}\xi(S)^{n-1}\mathrm{d}\xi(Q) \mathrm{d}\xi(S) \, , \label{omegapole}
\end{equation}
where $\Omega_{mn}(P)$ are holomorphic in $P$.
\item It is normalized such that
\begin{equation}
\oint_{\mathfrak{a}_j}\Omega(Q,S)=0,\qquad j =1, \ldots, g.\label{omeganorm}
\end{equation}
From \eqref{omeganorm} and the bilinear Riemann relation (see, e.g., \cite{bbeim94, ghmt08} for details) it follows that the $\mathfrak{b}$--periods of
$\Omega(Q,S)$ are
\begin{equation}
\oint_{\mathfrak{b}_j}\Omega(Q,S)=2\imath\pi \mathrm{d} v_j(S),\qquad j =1, \ldots, g \, , \label{bperiods}
\end{equation}
where $\mathrm{d} \boldsymbol{v} (S)= (\mathrm{d} v_1(S), \ldots, (\mathrm{d} v_g(S)  )^T  = (2\omega)^{-1} \mathrm{d} \boldsymbol{u}(S)$ is the vector of normalized holomorphic differentials. 
\end{enumerate}

\begin{figure}[t]
\begin{center}
\unitlength 0.6mm \linethickness{0.4pt}
\begin{picture}(150.00,80.00)
\put(-11.,33.){\psline{*-*}(0.5,0)(0,0)}
\put(-10.,29.){\makebox(0,0)[cc]{$e_1$}}
\put(0,29.){\makebox(0,0)[cc]{$e_2$}}
\put(-6.7,33.){\black\oval(20,30.)}
\put(-12.5,16){\makebox(0,0)[cc]{$\mathfrak{ a}_1$}}
\psline[arrowsize=2pt 3, linecolor=black]{->}(-0.45, 2.88)(-0.35, 2.88)
\put(12.,33.){\psline{*-*}(0.55,0)(0,0)}
\put(13.,29.){\makebox(0,0)[cc]{$e_3$}}
\put(22.,29.){\makebox(0,0)[cc]{$e_4$}}
\put(17.,33.){\black\oval(18.,26.)}
\put(10.,18.){\makebox(0,0)[cc]{$\mathfrak{ a}_2$}}
\psline[arrowsize=2pt 3, linecolor=black]{->}(0.873, 2.755)(0.98, 2.763)
\put(35.,33.){\circle*{1}} \put(40.,33.){\circle*{1}}
\put(45.,33.){\circle*{1}}
\put(60.,33.){\psline{*-*}(0.6,0)(0,0)} \put(60.,33.){\circle*{1}}
\put(69.,33.){\circle*{1}}
\put(60.,29.){\makebox(0,0)[cc]{$e_{2g-1}$}}
\put(72.,29.){\makebox(0,0)[cc]{$e_{2g}$}}
\put(65.,33.){\black\oval(30,15.0)}
\put(59.,20.){\makebox(0,0)[cc]{$\mathfrak{ a}_g$}}
\psline[arrowsize=2pt 3, linecolor=black]{->}(3.5, 2.43)(3.6, 2.43)
\put(114.,33.00){\psline{-*}(2.5,0)(0,0)} \put(114.,33.){\circle*{1}}
\put(115.,29.){\makebox(0,0)[cc]{$e_{2g+1}$}}
\put(146.,29.){\makebox(0,0)[cc]{$e_{2g+2}=\infty$}}
\put(35, 68.){\makebox(0,0)[cc]{$\mathfrak{ b}_1$}}
\bezier{484}(-7.,33.)(-7, 76.)(63.5, 76.)
\bezier{816}(63.5, 76.)(134., 76)(134., 33.)
\bezier{35}(-7.,33.)(-7., -7)(63.5, -7)
\bezier{35}(63.5, -7)(134., -7)(134.,33.)
\bezier{484}(17.,33.)(17.,60.)(72.,60.)
\bezier{816}(72.,60.)(127., 60)(127., 33.)
\bezier{35}(17.,33.)(17.,6)(72.,6)
\bezier{35}(72.,6)(127., 6)(127.,33.)
\put(90.,40.){\makebox(0,0)[cc]{$\mathfrak{b}_g$}}
\bezier{384}(65.,33.)(65.,45.)(92.5,45.)
\bezier{516}(92.5,45.)(120.,45.)(120.00,33.00)
\bezier{30}(65,33.00)(65,21.)(92.5, 21.)
\bezier{30}(92.5, 21.)(120,21.00)(120.00,33.00)
\put(63.,55.){\makebox(0,0)[cc]{$\mathfrak{ b}_2$}}
\psline[arrowsize=2pt 3, linecolor=black]{->}(2.33, 4.46)(2.46, 4.48)
\psline[arrowsize=2pt 3, linecolor=black]{->}(4.1, 3.6)(4.23, 3.6)
\psline[arrowsize=2pt 3, linecolor=black]{->}(5.6, 2.7)(5.73, 2.7)
%
\end{picture}
\end{center}
\caption{A homology basis on a Riemann surface of the hyperelliptic curve of genus $g$ with real branch points $e_1,\ldots,e_{2g+2}=\infty$ (upper sheet).  The cuts are drawn from $e_{2i-1}$ to $e_{2i}$ for $i=1,\dots,g+1$. The $\mathfrak b$-cycles are completed on the lower sheet.} \label{figure-1}
\end{figure}
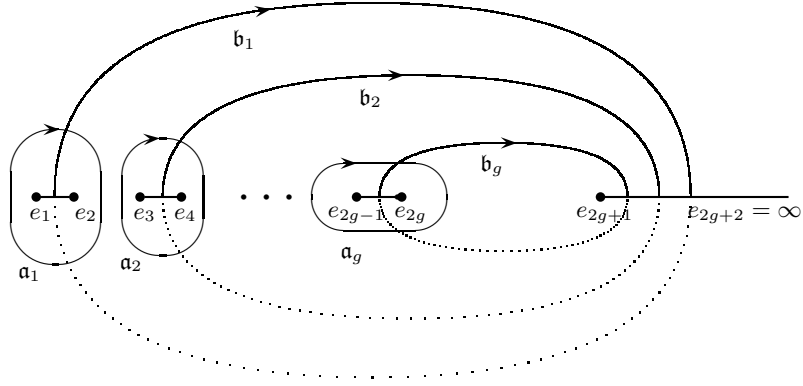

We present here the algebraic construction of the fundamental bi-differential $\Omega(Q,S)$. To do that we will construct at first a non-normalized bi-differential $\Gamma(Q,S)$ subject to the first two items in the definition of $\Omega(Q,S)$ above.

\begin{lemma} A symmetric bi-differential with the only second order pole along the diagonal defined up to a bilinear symmetric form in holomorphic differentials is given by
\begin{align}
\Gamma(P,Q)=\frac{\partial}{\partial z} \frac{y+w}{2(x-z)}\frac{\mathrm{d}x\mathrm{d}z}{y} +\mathrm{d}\boldsymbol{r}(z,w)^T \mathrm{d}\boldsymbol{u}(x,y) \, , \label{gammadiff}
\end{align}
where $P = (x, y)$ and $Q = (z, w)$, or, equivalently,
\begin{align}
\Gamma(P,Q)=\frac{F(x,y)+2yw}{4(x-z)^2} \frac{\mathrm{d}x}{y}\frac{\mathrm{d}z}{w} \ .\label{polara}
\end{align}
Here $\mathrm{d}\boldsymbol{r}$ is the vector of meromorphic differentials~\eqref{meromorphicdiff}
and $F(x,z)$ is a so-called Kleinian $2$-polar given by
\begin{equation}
F(x,z)=\sum^g_{k=0} x^k z^k \left( 2\lambda_{2k} + \lambda_{2k+1}(z+x) \right) \  \label{2polar}
\end{equation}
such that $F(x,x)=2y^2$ and $F(x,z)=F(z,x)$.
\end{lemma}

\begin{proof}
Consider the differential of the third kind~\eqref{thirdkind}
\begin{equation}
\Omega_{Q,Q'}(P)=\frac{y+w}{2(x-z)}\frac{\mathrm{d}x}{y}
-\frac{y+w'}{2(x-z')}\frac{\mathrm{d}x}{y}\label{thirdkind-2}
\end{equation}
depending on the variable $P=(x,y)$ and possessing poles in the points $Q=(z,w)$ and $Q'=(z',w')$. The bi-differential
\begin{equation}
\frac{\partial}{\partial z} \Omega_{Q,Q'}(P)\mathrm{d}z= \frac{\partial}{\partial z} \frac{y+w}{2(x-z)}\frac{\mathrm{d}x\mathrm{d}z}{y}
\end{equation}
as a form in $P$ has a second order pole along the diagonal $P=Q$, but as a form in $Q$ it has unwanted poles at $z=\infty$. This 2-form can be symmetrized ($\Gamma(P,Q)=\Gamma(Q,P)$) by adding an additional term $\mathrm{d}\boldsymbol{r}(z,w)^T \mathrm{d}\boldsymbol{u}(x,y)$ 
that annihilates the aforementioned poles. 
\end{proof}

The differential~\eqref{gammadiff} is defined up to a holomorphic 2-form  $2\mathrm{d}\boldsymbol{u}^T(z,w) \varkappa \mathrm{d}\boldsymbol{u}(x,y)$, where $\varkappa$ is a symmetric $g\times g$-matrix, $\varkappa^T=\varkappa$. This fact will be used for the symmetrization and normalization of the bi-differential. Thus, the bi-differential $\Gamma(P,Q)$ turns into
\begin{equation}
\Omega(P,Q)=\frac{\partial}{\partial z} \frac{y+w}{2(x-z)}\frac{\mathrm{d}x\mathrm{d}z}{y} +\mathrm{d}\boldsymbol{r}(z,w)^T \mathrm{d}\boldsymbol{u}(x,y)+2\mathrm{d}\boldsymbol{u}^T(z,w)\varkappa \mathrm{d}\boldsymbol{u}(x,y)\label{gammadiff1}
\end{equation}
or equivalently
\begin{equation}
\Omega(P,Q)=\frac{F(x,y)+2yw}{4(x-z)^2} \frac{\mathrm{d}x}{y}\frac{\mathrm{d}z}{w}+2\mathrm{d}\boldsymbol{u}^T(z,w)\varkappa \mathrm{d}\boldsymbol{u}(x,y) \ . \label{polara1}
\end{equation}
The matrix $\varkappa$ is chosen so that it normalizes $\Omega(P,Q)$ according to~\eqref{omeganorm} and~\eqref{bperiods} and the factor 2 in the second term of \eqref{polara1} is chosen to get precisely the Weierstra{\ss} definitions in the case $g=1$. This defines the matrix $\varkappa$ in terms of the $2\eta$-- and $2\omega$--periods as
\begin{equation}
\varkappa=\eta(2\omega)^{-1} \ . \label{varkappa} 
\end{equation}

We denote by $\mathrm{Jac}(X_g)$ the Jacobian of the curve $X_g$, i.e., the factor $\mathbb{C}^g/ \Gamma$, where $\Gamma=2 \omega \oplus 2 \omega'$ is the lattice generated by the periods of the canonical holomorphic differentials. Any point $\boldsymbol{u} \in \mathrm{Jac}(X_g)$ can be represented in the form
\begin{equation}
\boldsymbol{u}=2\omega \boldsymbol{\varepsilon}+2\omega' \boldsymbol{\varepsilon}' \, ,
\end{equation}
where $ \boldsymbol{\varepsilon},\boldsymbol{\varepsilon}'\in \mathbb{R}^g $. 
The vectors $\varepsilon$ and $\varepsilon^\prime$ combine to a $2 \times g$ matrix and form the characteristic $\varepsilon$ of the point $\boldsymbol{u}$,
\begin{equation}
[\boldsymbol{u}] := \begin{pmatrix} \boldsymbol{\varepsilon'}^T \\
\boldsymbol{\varepsilon}^T \end{pmatrix} = \begin{pmatrix} \varepsilon_1' & \ldots & \varepsilon_g' \\
\varepsilon_1 & \ldots & \varepsilon_g \end{pmatrix} =: \varepsilon \, .
\end{equation}
If $\boldsymbol{u}$ is a half-period, then all entries of the characteristic $\varepsilon$ are equal to $\frac12$ or $0$.

Beside the canonic holomorphic differentials $\mathrm{d}\boldsymbol{u}$ we will also consider the normalized holomorphic differentials defined by
\begin{equation}
 \label{eq:normdiff}
 \d \vek{v} = (2\omega)^{-1}\d \vek{u} \, .
\end{equation}
Their corresponding holomorphic periods are $1_g$ and $\tau$, where the Riemann period matrix $\tau := \omega^{-1} \omega'$ is in the Siegel upper half space $\mathfrak{S}_g$ of $g \times g$--matrices (or half space of degree $g$),
\begin{equation}
\mathfrak{S}_g = \left\{\tau \;\; g \times g \;\; \text{matrix} \big| \tau^T=\tau, \, \mathrm{Im}(\tau) \,\, \text{positive definite} \right\}\,.
\end{equation}
The corresponding Jacobian is introduced as
\begin{equation}
\widetilde{\mathrm{ Jac}}(X_g) := (2\omega)^{-1}\mathrm{ Jac} (X_g)=\mathbb{C}^g/ 1_g\oplus \tau \,.
\end{equation}
We will use both versions: the first one $(2\omega,2\omega')$ in the context of the $\sigma$--functions, and the second one $(1_g,\tau)$ in the case of the $\theta$--functions.

The Abel map $\boldsymbol{\mathfrak{A}}: (X_g)^n \rightarrow \mathbb{C}^g$ with the base point $P_0$ relates the set of points $(P_1,\ldots,P_N)$ which are called the divisor $\mathcal{D}$, with a point in the Jacobian $\mathrm{ Jac}(X_g)$
\begin{equation}
\boldsymbol{\mathfrak{A}}(P_1,\ldots,P_N) := \sum_{k=1}^N  \int_{P_0}^{P_k} \d \boldsymbol{u} \,. \label{jip}
\end{equation}
The divisor $\mathcal{D}$ in~\eqref{jip} can also be denoted as $P_1+\ldots+P_N - N P_0$.

More generally the divisor $\mathcal{D}$ is the formal sum $\mathcal{D}= n_1P_1+\ldots+n_NP_N$ with integers $n_j$, $j=1,\ldots,N$ and $N\in \mathbb{N}$. The degree of the divisor,
 $\mathrm{deg}(\mathcal{D})$, is the sum $\mathrm{deg}(\mathcal{D})=n_1+\ldots+n_N$. The divisor of a meromorphic function is of degree zero. Two divisors  $\mathcal{D}$ and $\mathcal{D}'$ are linearly equivalent if $\mathcal{D}-\mathcal{D}'$ is the divisor of a meromorphic function. Linearly equivalent divisors constitute a class. In particular, the canonical class $\mathcal{K}_{X_g}$ is the divisor class of abelian differentials of degree $2g-2$.

The divisor is positive if all $n_j\geq0$. Let $l(\mathcal{D})$ be the dimension of the space of meromorphic functions that have poles in the points of $\mathcal{D}$ of multiplicities not higher than the multiplicity of these points in $\mathcal{D}$. If $\mathrm{deg}(\mathcal{D})\geq g$ for a divisor in general position the dimension $l(\mathcal{D})$ is given by
\begin{equation} 
l(\mathcal{D})= \mathrm{deg}(\mathcal{D})-g +1 \ . 
\end{equation}
Such divisors are called non-special. All remaining divisors with $\mathrm{deg}(\mathcal{D})\geq2$ are called special. For a detailed explanation see e.g. \cite{fk80}.

Analogously we define
\begin{equation}
\widetilde{\boldsymbol{\mathfrak{A}}}(P_1,\ldots,P_n)= \sum_{k=1}^n  \int_{P_0}^{P_k} \d \boldsymbol{v}= (2\omega)^{-1}\boldsymbol{\mathfrak{A}}(P_1,\ldots,P_n)\,.
 \label{jiptilde}
\end{equation}
In the context of our consideration we choose $P_0$ at infinity, $P_0=(\infty,\infty)$.

\subsection{$\theta$-- and $\sigma$--functions}

The hyperelliptic $\theta$--function with characteristic $\varepsilon$ is a mapping $\theta:\ \widetilde{\mathrm{Jac}}(X_g) \times \mathfrak{S}_g \rightarrow \mathbb{C}$ defined through the Fourier series
\begin{equation}
\theta[\varepsilon] (\vek{ v}| \tau) := \sum_{\vek{m} \in \mathbb{Z}^g} e^{\pi \imath \left\{(\vek{m} + \vek{\varepsilon}')^T \tau (\vek{m} + \vek{\varepsilon}') + 2 (\vek{ v}+ \vek{\varepsilon})^T (\vek{m} + \vek{\varepsilon}')\right\}} \, .
\end{equation}
It possesses the periodicity property
\begin{equation}
\theta[\varepsilon] (\vek{v}+\vek{n} + \tau\vek{n}'|\tau) = e^{-2\imath\pi{\vek{n}'}^T(\vek{v} + \frac12\tau \vek{n}')} e^{2\imath\pi(\vek{n}^T\vek{\varepsilon}' - {\vek{n}'}^T\vek{\varepsilon})} \theta[\varepsilon](\vek{v}|\tau) \, . \label{thetachar2a}
\end{equation}
For vanishing characteristic we abbreviate $\theta(\vek{v}) := \theta[0] (\vek{v}| \tau)$.

In the following, the values $\varepsilon_k$, $\varepsilon_k'$ are either $0$ or $\frac{1}{2}$. The property~\eqref{thetachar2a} implies
\begin{equation}
\theta[\varepsilon]({-\vek v}|\tau)= \mathrm{e}^{-4\pi \imath
\vek{\varepsilon}^T\vek{\varepsilon}'} \theta[\varepsilon](\vek{v}|\tau)\label{-v},
\end{equation}
so that the function $\theta[\varepsilon](\vek{v}|\tau)$ with characteristic $\varepsilon$ of only half-integers is even if $4\vek{\varepsilon}^T\vek{\varepsilon}'$ is an even integer, and odd otherwise. Correspondingly, $\varepsilon$ is called even or odd, and among the $4^g$ half-integer characteristics there are $\frac{1}{2} (4^g+2^g) $ even and $\frac{1}{2}(4^g-2^g) $ odd characteristics.

The nonvanishing values of the $\theta$-functions with half--integer characteristics and their derivatives are called $\theta$-constants and are denoted as
\begin{align*}
\theta[\varepsilon] & : =\theta[\varepsilon](\vek{0};\tau), & \quad  \theta_{ij}[\varepsilon] & := \left.\frac{\partial^2}{\partial z_i\partial z_j}\theta[\varepsilon](\boldsymbol{z};\tau)\right|_{\boldsymbol{z}=0},\quad & \text{etc.} \qquad \text{for even $[\varepsilon]$};\\
\theta_{i}[\varepsilon] & :=\left.\frac{\partial}{\partial z_i}\theta[\varepsilon](\boldsymbol{z};\tau)\right|_{\boldsymbol{z}=0}, & \quad
\theta_{ijk}[\varepsilon] & := \left.\frac{\partial^3}{\partial z_i\partial z_j \partial z_k}\theta[\varepsilon](\boldsymbol{z};\tau)\right|_{\boldsymbol{z}=0},\quad
  & \text{etc.} \qquad \text{for odd $[\varepsilon]$} \,.
\end{align*}
Even characteristics $\varepsilon$ are called nonsingular if $\theta[\varepsilon]\neq 0$, and odd characteristics $\varepsilon$ are called nonsingular if $\theta_i[\varepsilon]\neq 0$ for at least one index $i$.

We identify each branch point $e_j$ of the curve $X_g$ with a vector
\begin{equation}
{\boldsymbol{\mathfrak{A}}}_j := \int_{\infty}^{(e_j,0)} \d \vek{u} =: 2 \omega \boldsymbol{\varepsilon}_j + 2 \omega' \boldsymbol{\varepsilon}'_j \in \mathrm{Jac}(X_g),\quad j=1,\ldots, 2g+2 \, ,
\end{equation}
which defines the two vectors $\boldsymbol{\varepsilon}_j$ and $\boldsymbol{\varepsilon}'_j$. Evidently, $[{\boldsymbol{\mathfrak A}}_{2g+2}] = [0] = 0$.

In terms of the $2g+2$ characteristics $[{\mathfrak A}_i]$ all $4^g$ half integer characteristics $\varepsilon$  can be constructed as follows. There is a one-to-one correspondence between these $\varepsilon$ and the partitions of the set $\bar{{\mathcal G}} = \{1, \ldots, 2g+2\}$ of indices of the branch points~(\cite{fa73}, p.~13, \cite{ba97} p.~271). The partitions of interest are
\begin{equation}
 \label{partitionsm}
{\mathcal I}_m \cup {\mathcal J}_m = \{ i_1, \ldots, i_{g+1-2m}\}\cup \{ j_1, \ldots, j_{g+1+2m}\},
\end{equation}
where $m$ is any integer between $0$ and $\left[\frac{g+1}{2}\right]$. The corresponding characteristic $\boldsymbol{\varepsilon}_{m}$ is defined by the vector
\begin{equation}
\boldsymbol{\Delta}_m = \sum_{k=1}^{g+1-2m}\widetilde{\boldsymbol{\mathfrak{A}}}_{i_k}
 +\vek{ K}_{\infty} =: \boldsymbol{\varepsilon}_m +\tau \boldsymbol{\varepsilon}'_m \, ,
\label{sumchar}
\end{equation}
where $\boldsymbol{K}_{\infty}\in \widetilde{\mathrm{Jac}}(X_g)$ is the vector of Riemann constants with base point $\infty$, which will always be used in the argument of the $\theta$-functions, and which is given as a vector in $\widetilde{\mathrm{Jac}}(X_g)$ by
\begin{equation}
\boldsymbol{K}_\infty := \sum_{\text{all odd}\; [\mathfrak{A}_j]} \widetilde{\boldsymbol{\mathfrak A}}_j
\label{rvectorgen}
\end{equation}
(see e.g. \cite{fk80}, p. 305, for a proof).

It can be seen that characteristics with even $m$ are even, and with odd $m$ are odd. There are $\frac{1}{2}{2g+2 \choose g+1}$ different partitions with $m=0$, ${2g+2 \choose g-1}$ different partitions with $m=1$, and, in general, ${2g+2 \choose g+1-2m}$ down to ${2g+2 \choose 1}=2g+2$ partitions if $g$ is even and $m=g/2$, or ${2g+2 \choose 0}=1$ partitions if $g$ is odd and $m=(g+1)/2$.  One may check that the total number of even (odd) characteristics is indeed $2^{2g-1}\pm 2^{g-1}$. According to the Riemann theorem for the zeros of $\theta$-functions~\cite{fa73}, $\theta(\boldsymbol{\Delta}_m+\vek{v})$ vanishes to order $m$ at $\vek{v}=0$ and in particular, the function $\theta(\boldsymbol{K}_{\infty}+\vek{v})$ vanishes to order $\left[\frac{g+1}{2}\right]$ at $\vek{v}=0$.

Let us demonstrate, following \cite{fk80}, p.~303, how the set of characteristics $[\boldsymbol{\mathfrak{A}}_k] \equiv [\widetilde{\boldsymbol{\mathfrak{A}}}_k]$, $k=1,\ldots, 2g+2$ looks like in the homology basis shown in Figure~\ref{figure-1}.
 Using the notation $\vek{f}_k = \frac{1}{2}(\delta_{1k}, \ldots, \delta_{gk})^t$ and $\boldsymbol{\tau}_k$
for the $k$-th column vector of the matrix $\tau$, we find
\begin{align}
\widetilde{\boldsymbol{\mathfrak A}}_{2g+1} & 
= \widetilde{\boldsymbol{\mathfrak A}}_{2g+2} - \sum_{k=1}^g\int\limits_{(e_{2k-1},0)}^{(e_{2k},0)}\d \vek{v}
 = \sum_{k=1}^g\vek{f}_k, &  \to  & & [\widetilde{{\boldsymbol{\mathfrak A}}}_{2g+1}] & = \frac12 \begin{pmatrix} 0 & 0 & \ldots & 0 & 0 \\ 1 & 1 & \ldots & 1 & 1 \end{pmatrix}, \nonumber \\
\widetilde{\boldsymbol{\mathfrak A}}_{2g} & = \widetilde{\boldsymbol{\mathfrak A}}_{2g+1} - \int\limits_{(e_{2g+1},0)}^{(e_{2g},0)}\d \vek{v} = \sum_{k=1}^g\vek{f}_k + \vek{\tau}_g, &  \to  & & [\widetilde{{\boldsymbol{\mathfrak A}}}_{2g}] & = 
\frac12 \begin{pmatrix} 0 & 0 & \ldots & 0 & 1 \\ 1 & 1 & \ldots & 1 & 1 \end{pmatrix}, \\
\widetilde{\boldsymbol{\mathfrak A}}_{2g-1} & = \widetilde{\boldsymbol{\mathfrak A}}_{2g} - \int\limits_{(e_{2k-1},0)}^{(e_{2k},0)}\d \vek{v} = \sum_{k=1}^{g-1}\vek{f}_k + \vek{\tau}_g, &  \to  & & [\widetilde{{\boldsymbol{\mathfrak A}}}_{2g-1}] & = \frac12 \begin{pmatrix} 0 & 0 & \ldots & 0 & 1 \\ 1 & 1 & \ldots & 1 & 0 \end{pmatrix} \, . \nonumber
\end{align}
Continuing in the same manner, we get for arbitrary $1 \leq k < g$
\begin{align}
\begin{split}
[\widetilde{{\boldsymbol{\mathfrak A}}}_{2k+2}] & = \frac12 \Biggl(\overbrace{\begin{matrix} 0 & 0 & \ldots & 0 \\ 1 & 1 & \ldots & 1 \end{matrix}}^{k} \;\;
\begin{matrix} 1 & 0 & \ldots & 0 \\ 1 & 0 & \ldots & 0 \end{matrix}\Biggr), \\
[\widetilde{{\boldsymbol{\mathfrak A}}}_{2k+1}] & = \frac12 \Biggl(\overbrace{\begin{matrix} 0 & 0 & \ldots & 0 \\ 1 & 1 & \ldots & 1 \end{matrix}}^{k} \;\;
\begin{matrix} 1 & 0 & \ldots & 0 \\ 0 & 0 & \ldots & 0 \end{matrix}\Biggr)
\end{split}
\end{align}
and finally
\begin{equation}
[\widetilde{\boldsymbol{{\mathfrak A}}}_2] = \frac12 \begin{pmatrix} 1 & 0 & \ldots & 0 \\ 1 & 0 & \ldots & 0 \end{pmatrix} \, , \qquad
[\widetilde{\boldsymbol{{\mathfrak A}}}_1] = \frac12 \begin{pmatrix} 1 & 0 & \ldots & 0 \\ 0 & 0 & \ldots & 0
\end{pmatrix}\, .
\end{equation}
The characteristics with even indices, corresponding to the branch points $e_{2n}$, $n=1,\ldots, g$, are odd (except for $[{\mathfrak A}_{2g+2}]$ which is zero); the others are even.
Therefore in the basis drawn in Figure \ref{figure-1} we get
\begin{equation}
\boldsymbol{K}_\infty =  \sum_{k=1 }^g \widetilde{\boldsymbol{\mathfrak A}}_{2k} \, .
\label{rvector}
\end{equation}
The formula (\ref{rvector}) is in accordance with the classical theory where the vector of Riemann constants is defined as (see Fay \cite{fa73}, Eq. (14))
\begin{equation}
\mathrm{Divisor} \, \boldsymbol{K}_{P_0} = \Delta-(g-1)P_0 \, ,
\end{equation}
where $\Delta$ is the divisor of degree $g-1$ that is the {\em{Riemann divisor}}. In the case considered $P_0=\infty$ and $\Delta = e_{2} + e_4 + \ldots + e_{2g} -\infty$. The calculation of the divisor of the differential $\prod_{k=1}^g(x-e_{2k}) \mathrm{d}x/y$ leads to the required conclusion $2 \Delta = \mathcal{K}_{X_g}$ where  $\mathcal{K}_{X_g}$ is the canonical class.

The Kleinian $\sigma$--function of the hyperelliptic curve $X_g$ is defined over the Jacobian $\mathrm{Jac}(X_g)$ as
\begin{equation}
\sigma(\boldsymbol{u};M) := C \theta[\boldsymbol{K}_{\infty}]((2\omega)^{-1} \boldsymbol{u} ;\tau) \, e^{\boldsymbol{u}^T \varkappa \boldsymbol{u}} \label{sigma} \ ,
\end{equation}
where the symmetric $g\times g$ matrix $\varkappa$ is defined in~\eqref{varkappa}.
Here $\boldsymbol{K}_{\infty}\in \mathrm{Jac}(X_g)$ and
\begin{equation}
\boldsymbol{u}= \int_{g\infty}^{\mathcal{D}} \mathrm{d}\boldsymbol{u} \equiv
\sum_{k=1}^g \int_{g\infty}^{P_k}\mathrm{d}\boldsymbol{u},
\end{equation}
where $\mathcal{D}=P_1+\ldots+P_g$ is a divisor in the general position. The constant
\begin{equation}
C = \sqrt{ \frac{\pi^g}{\mathrm{det}(2\omega)} }\left(\prod_{1\leq i <j \leq 2g+1} (e_i-e_j)\right)^{-1/4} \,, \label{sigmac}
\end{equation}
and $M$ defined in \eqref{defM} contains the set of all moduli $2\omega,2\omega'$ and $2\eta, 2\eta'$. In the following we will use the shorter notation $\sigma(\boldsymbol{u};M)=\sigma(\boldsymbol{u})$. Sometimes the $\sigma$-function~\eqref{sigma} is called fundamental $\sigma$-function.

The multi-variable $\sigma$-function~\eqref{sigma} represents a natural generalization of the Weierstra{\ss}
$\sigma$-function given by
\begin{equation}
\sigma(u)= \sqrt{\frac{\pi}{2\omega}} \frac{\epsilon}{\sqrt[4]{(e_1-e_2)(e_1-e_3)(e_2-e_3)}} \vartheta_1\left( \frac{u}{2\omega}\right) \mathrm{exp} \left\{ \frac{\eta u^2}{2\omega} \right\},\quad \epsilon^8=1\, , \label{fundamental}
\end{equation}
where $\vartheta_1$ is the standard $\theta$-function. We note that~\eqref{sigma} differs in the case of genus one from the Weierstra{\ss} $\sigma$--function by an exponential factor that appears when the shift on a half period in the $\theta$-argument is taken into account in the $\theta$-characteristics.

The fundamental $\sigma$--function \eqref{sigma} possesses the properties
\begin{itemize}
\item It is an entire function on  $\mathrm{Jac}(X_g)$,
\item It satisfies the two sets of functional equations
\begin{equation}
\begin{split}
\sigma(\boldsymbol{u}+2\omega \boldsymbol{k}+2\omega'\boldsymbol{k}^\prime; M) & = e^{2 (\eta\boldsymbol{k}+\eta'\boldsymbol{k}')^T (\boldsymbol{u}+\omega\boldsymbol{k}+\omega'\boldsymbol{k}')} \sigma(\boldsymbol{u};M)\\
\sigma(\boldsymbol{u};(\gamma M^T)^T) & = \sigma(\boldsymbol{u};M) \, , \end{split}
\end{equation}
where $\gamma \in \mathrm{Sp}(2g,\mathbb Z)$, that is, $\gamma J \gamma^{T} = J$, and $M^T$ is the matrix $M$ with interchanged submatrices $\omega^\prime$ and $\eta$.
The first of these equations displays the {\it periodicity property}, and the second one the {\it modular property}.


\item In the vicinity of the origin the power series of $\sigma(\boldsymbol{u})$ is of the form
\begin{equation}
\sigma(\boldsymbol{u}) = S_{\boldsymbol{\pi}}(\boldsymbol{u})+\text{higher order terms} \, ,  \label{sigmaseries}
\end{equation}
where $S_{\boldsymbol{\pi}}(\boldsymbol{u})$ are the Schur--Weierstra{\ss} functions associated to the curve $X_g$ and defined on $\mathbb{C}^g\ni(u_1,\ldots,u_g)$ by the Weierstra{\ss} gap sequence at the infinite branch point. For $g>1$ it is always a Weierstra{\ss} point. The partition ${\boldsymbol{\pi}}=(\pi_g,\ldots,\pi_1)$ is defined by the Weierstra{\ss} gap sequence ${\boldsymbol{w}}=(w_1,\ldots,w_g)$ as follows: $\pi_i=w_{g-i+1}+i-g$.
Details of the definition are given in \cite{BEL99}, see also \cite{ehkkl11}. As an example we will present here the first few functions $S_{\boldsymbol{\pi}}(\boldsymbol{u})$
\begin{align}
g & = 1: & \quad S_{1}(u_1) & = u_1,\label{schur1} \\
g & = 2: & \quad S_{2,1}(u_1,u_2) & = \frac13 u_2^3-u_1 , \label{schur2} \\
g & = 3: & \quad  S_{3,2,1}(u_1,u_2,u_3) & = \frac{1}{45}u_3^6 - \frac13 u_2 u_3^3 - u_2^2+u_1u_3,\label{schur3}\\
g & = 4: & \quad S_{4,3,2,1}(u_1,u_2,u_3,u_4) & = \frac{1}{4725} u_4^{10} - \frac{1}{105} u_4^7 u_3 + \frac{1}{15} u_2u_4^5 - u_4 u_3^3 - \frac{1}{3} u_4^3 u_1 \label{schur4}\\
& & & \qquad + u_2 u_3 u_4^2 - u_2^2 + u_1 u_3 \, . \nonumber
\end{align}
\end{itemize}
The partitions constructed by the Weierstra{\ss} gap sequences are denoted in the subscripts. In particular, in the case of genus $g=4$ the partition $\boldsymbol{\pi}=(1,2,3,4)$
corresponds to the gap sequence $\overline{0},\;1,\;\overline{2},\;3,\;\overline{4},\;5,\;\overline{6},\;7,\;\overline{8,\;9,\;10,\ldots}$, where orders of existing functions are overlined. These are so-called non-gap numbers, in contrast to the gap-numbers. The genus is defined by the number of gaps or, equivalently, the first number starting from which no gap appears equals $2g$ (in this example $8=2g$).

The Kleinian $\zeta$ and $\wp$-functions are a natural generalization of the Weierstra{\ss} $\zeta$ and $\wp$-functions and are given by the logarithmic derivatives of $\sigma$,
\begin{align}
\begin{split}
\zeta_{i}(\boldsymbol{u}) & = \frac{\partial}{ \partial u_i} \;
\mathrm{ln}\,\sigma(\boldsymbol{u}),\\
\wp_{ij}(\boldsymbol{u}) & = -\frac{\partial^2}{ \partial u_i\partial u_j} \;\mathrm{ln}\,\sigma(\boldsymbol{u}), \\
\wp_{ijk}(\boldsymbol{u}) & = -\frac{\partial^3}{ \partial u_i\partial u_j\partial u_k} \;\mathrm{ln}\,\sigma(\boldsymbol{u}) \,, \quad \text{etc.,}
\end{split}
\end{align}
where $i,j,k \in \{1,\ldots,g\}$. In this notation the Weierstra{\ss} $\wp$-function is $\wp_{11}(u)$. For convenience, we introduce the vector of $\zeta$-functions  $\boldsymbol{\zeta}(\boldsymbol{u})=(\zeta_1(\boldsymbol{u}),\ldots,\zeta_g(\boldsymbol{u}))^T$ and also denote the derivatives of the $\sigma$--function by
\begin{equation}
\sigma_i(\boldsymbol{u})=\frac{\partial}{\partial u_i} \sigma(\boldsymbol{u}),\quad \sigma_{ij}(\boldsymbol{u})=\frac{\partial^2}{\partial u_i\partial u_j} \sigma(\boldsymbol{u}), \qquad \text{etc.}
\end{equation}

\subsection{Main formula}\label{sec:mainformula}

We consider now the integration of the differentials of the second and third kind.

\begin{proposition} Let $X_g$ be a hyperelliptic curve of genus $g$ with branch point
at infinity. Let $\mathcal{D}=( P_1,\ldots P_g  )$, $\mathcal{D}'=( P_1',\ldots P_g'  )$, $P_k=(Z_k,W_k)$, $P_k'=(Z_k',W_k')$ be non-special divisors of degree $g$. Let $P=(x,y)$ and
$P'=(x',y')$ be two arbitrary points of $X$. Then
\begin{equation}
\int_{P'}^P \sum_{k=1}^g \int_{P_k'}^{P_k} \frac{F(x,z)+2yw}{4(x-z)^2}  \frac{\mathrm{d}x}{y} \frac{\mathrm{d}z}{w}
=\mathrm{ln} \frac{\sigma\left(\int_{P_0}^P \mathrm{d}\boldsymbol{u}- \int_{g\infty}^{\mathcal{D}} \mathrm{d}\boldsymbol{u}  \right)}{\sigma\left(\int_{P_0}^P \mathrm{d}\boldsymbol{u}- \int_{g\infty}^{\mathcal{D}'} \mathrm{d}\boldsymbol{u}  \right)}
-\mathrm{ln} \frac{\sigma\left(\int_{P_0}^{P'} \mathrm{d}\boldsymbol{u}- \int_{g\infty}^{\mathcal{D}} \mathrm{d}\boldsymbol{u}  \right)}{\sigma\left(\int_{P_0}^{P'} \mathrm{d}\boldsymbol{u}- \int_{g\infty}^{\mathcal{D}'} \mathrm{d}\boldsymbol{u}  \right)} \label{main}  \ .
\end{equation}
\end{proposition}

\begin{proof} We introduce non-special divisors
\[ \mathcal{D}= ( (Z_1,W_1),\ldots, (Z_g,W_g) ),\quad \mathcal{D}'= ( (Z_1',W_1'),\ldots, (Z_g',W_g') ) \]
as well as two arbitrary points $P=(x,y)$, $P'=(x',y')$. The integration of $\Omega(P,Q)$ given by~\eqref{gammadiff1}
\begin{equation}
\int_{P'}^P \sum_{k=1}^g  \int_{(Z'_k,W'_k)}^{(Z_k,W_k)} \Omega(P,Q) \,
\end{equation}
yields, according to the Riemann vanishing theorem,
\begin{equation}
\mathrm{ln} \frac{
\theta\left( \int_{P_0}^{P} \mathrm{d}\boldsymbol{v}-\sum_{k=1}^g \int_{P_0}^{(Z_k,W_k)} \mathrm{d}\boldsymbol{v} +\boldsymbol{K}_{P_0}   \right)}
{\theta\left( \int_{P_0}^{P} \mathrm{d}\boldsymbol{v}-\sum_{k=1}^g \int_{P_0}^{(Z_k',W_k')} \mathrm{d}\boldsymbol{v} +\boldsymbol{K}_{P_0}   \right)}-
\mathrm{ln} \frac{
\theta\left( \int_{P_0}^{P'} \mathrm{d}\boldsymbol{v}-\sum_{k=1}^g \int_{P_0}^{(Z_k,W_k)} \mathrm{d}\boldsymbol{v} +\boldsymbol{K}_{P_0}   \right)}
{
\theta\left( \int_{P_0}^{P'} \mathrm{d}\boldsymbol{v}-\sum_{k=1}^g \int_{P_0}^{(Z_k',W_k')} \mathrm{d}\boldsymbol{v} +\boldsymbol{K}_{P_0}   \right)} \ , \label{mainformula}
\end{equation}
where $P_0$ is the base point of the Abel map (that we suppose to be infinity) and $\boldsymbol{K}_{P_0}$ is the vector of Riemann constants with the base point $P_0$. Using the definition of the fundamental $\sigma$-function~\eqref{sigma} we get~\eqref{main}.
\end{proof}

The following corollaries follow from the main formula~\eqref{main}.

\begin{corollary}
Let $P=(x,y)$, $P_k=(x_k,y_k), k=1,\ldots,g$. Then
\begin{equation}
\sum_{ i, j=1 }^g \wp_{ij} \left(\int_{P_0}^{P}\mathrm{d}\boldsymbol{u} - \sum_{r=1}^{g}\int_{P_0}^{P_k}\mathrm{d}\boldsymbol{u} \right)  x_k^{i-1} x^{j-1}=  \frac{F(x,x_k) + 2yy_k}{4(x-x_k)^2},\quad k=1,\ldots,g \label{kleinformula}
\end{equation}
\end{corollary}
\begin{proof}
Take the partial derivative $\partial^2/\partial{x}\partial{x_k}$ on both sides of~\eqref{main}.
\end{proof}

\begin{corollary}
Let $\mathcal{D}=P_1+\ldots+P_g$ be a positive divisor of degree $g$. Then the standard Jacobi inversion problem given by the equations
\begin{equation}
\sum_{k=1}^g  \int_{\infty}^{P_k} \mathrm{d}\boldsymbol{u}=\boldsymbol{u}  \label{JIPjolomorphic} \end{equation}
is solved in terms of Kleinian $\wp$-functions as\footnote{In the previous work~\cite{ehkkl11} in this formula numbered as (3.44) as well in its particular cases (5.5) and (6.5) the sign ``-"
was misplaced.}
\begin{eqnarray}
x^g-\wp_{gg}(\boldsymbol{u})x^{g-1}-\wp_{g,g-1}(\boldsymbol{u})x^{g-2}-\ldots - \wp_{g,1}(\boldsymbol{u})=0 \ , && \label{JIP1} \\
y_k=\wp_{ggg}(\boldsymbol{u})x_k^{g-1}+\wp_{gg,g-1}(\boldsymbol{u})x_k^{g-2}+\ldots + \wp_{gg,1}(\boldsymbol{u}),\qquad k=1,\ldots,g \ . && \label{JIP2}
\end{eqnarray}
\end{corollary}
\begin{proof}
To prove~\eqref{JIP1} consider $x\rightarrow\infty$. Substitute $x=1/\xi^2$ into~\eqref{kleinformula}. Comparison of the coefficients of $1/\xi^{2(g-1)}$ on the RHS and LHS of~\eqref{kleinformula} gives relation~\eqref{JIP1}. Formula~\eqref{JIP2} follows from~\eqref{JIP1} and~\eqref{kleinformula}. For details see~\cite{BEL97}.
\end{proof}

For the meromorphic differentials of the second kind the following relation was proved by Buchstaber and Leykin~\cite{BL05}:
\begin{lemma}[\bf $\zeta$-formula]\label{lemma:zeta-formula}
Let $\mathcal{D}_0$ be a divisor supported by $g$ branch points $e_{i_1},\ldots, e_{i_g}$ such that
\begin{equation}
\sum_{k=1}^g \int_{\infty}^{(e_{i_k},0)} \mathrm{d} \boldsymbol{u}= \boldsymbol{K}_{\infty} \ , \,\, [\boldsymbol{K}_{\infty}] := \begin{pmatrix} \boldsymbol{\varepsilon'}^T \\
\boldsymbol{\varepsilon}^T \end{pmatrix}
\end{equation}
and $\mathcal{D}$ is a non-special divisor of degree $g$, $\mathcal{D}=P_1+\ldots+P_g$.
Then for any vector
\begin{equation}
\boldsymbol{u}=\int_{g\infty}^{\mathcal{D}} \mathrm{d} \boldsymbol{u} \in \mathrm{Jac}(X)
\end{equation}
the following relation is valid
\begin{equation}
\int_{\mathcal{D}_0}^{\mathcal{D}} \mathrm{d}\boldsymbol{r}=-\boldsymbol{\zeta}(\boldsymbol{u}) + 2( \boldsymbol{\eta}^{\prime}\boldsymbol{\varepsilon}^\prime + \boldsymbol{\eta}\boldsymbol{\varepsilon} )
+\frac12\boldsymbol{\mathfrak{Z}} (\boldsymbol{u}),\label{JIPmeromorphic}  \end{equation}
where the components $\mathfrak{Z}_j(\boldsymbol{u})$ of the vector $\boldsymbol{\mathfrak{Z}}(\boldsymbol{u})$ are $$\mathfrak{Z}_g(\boldsymbol{u})=0,\qquad \mathfrak{Z}_{g-1}(\boldsymbol{u})=\wp_{ggg}(\boldsymbol{u})$$
and the other components at $1\leq j<g-1$ are given by the $j\times j$ determinants
\begin{equation}
\mathfrak{Z}_j(\boldsymbol{u})=\left|  \begin{array}{cccccc}  \wp_{gg}(\boldsymbol{u})&-1&0&0&\ldots&0\\
2\wp_{g-1,g}(\boldsymbol{u})&\wp_{gg}(\boldsymbol{u})&-1&0&\ldots&0\\
\ldots&\ldots&\ldots&\ldots&\ldots&\ldots\\
(g-k)\wp_{k+1,g}(\boldsymbol{u})&\wp_{k+2,g}(\boldsymbol{u})&\ldots&\ldots&\ldots&\ldots\\
\ldots&\ldots&\ldots&\ldots&\ldots&\ldots\\
(g-j-1)\wp_{j+2,g}(\boldsymbol{u})&\wp_{j+3,g}(\boldsymbol{u})&\ldots&\ldots&\wp_{gg}(\boldsymbol{u})&-1\\
(g-j)\wp_{j+1,g,g}(\boldsymbol{u})&\wp_{j+2,g,g}(\boldsymbol{u})&
\ldots&\ldots&\wp_{g-1,g,g}(\boldsymbol{u})&\wp_{ggg}(\boldsymbol{u})
\end{array}\right|\label{zwp}
\end{equation}
\end{lemma}
The corresponding formula in Buchstaber and Leykin~\cite{BL05} coincides with those given by Baker~\cite{ba97}, p.321, only for $j=g$ and $j=g-1$ and differs for $1\leq j<g-1$. Our further consideration is based on~\eqref{zwp}.

For our purposes it is necessary to rewrite the expressions for $\mathfrak{Z}_j$ in terms of symmetric functions of the divisor points $P_1=(x_1,y_1),\ldots,P_g=(x_g,y_g)$.
With the solution of the Jacobi inversion problem~\eqref{JIP1}-\eqref{JIP2} the components $\mathfrak{Z}_j(\boldsymbol{u})$ yield:
\begin{align*}
\mathfrak{Z}_g&=0\\
\mathfrak{Z}_{g-1}&=\frac{y_1}{(x_1-x_2)\cdots (x_1-x_g)}+\text{permutations},\\
\mathfrak{Z}_{g-2}&=y_1\frac{x_1-(x_2+\ldots+x_g)}{(x_1-x_2)\cdots (x_1-x_g)}+\text{permutations},\\
\mathfrak{Z}_{g-3}&=y_1\frac{x_1^2-x_1(x_2+\ldots+x_g)+x_2x_3+\ldots+x_{g-1}x_g }{(x_1-x_2)\cdots (x_1-x_g)}+\text{permutations} \ . \\
&\vdots
\end{align*}

\noindent
Note that the characteristics in the formula~\eqref{JIPmeromorphic} are not reduced.

From~\eqref{JIPmeromorphic} one obtains the relation between the periods of holomorphic and meromorphic integrals
\begin{equation}
{\eta}_{ik}= \left(   \zeta_i( \boldsymbol{\omega}_k + \boldsymbol{K}_{\infty}  )   \right)_{i,k=1,\ldots,g},\quad
{\eta}'_{ik}= \left(   \zeta_i( \boldsymbol{\omega}_k' + \boldsymbol{K}_{\infty}  )   \right)_{i,k=1,\ldots,g}
\end{equation}
with $2\boldsymbol{\omega}_k$ being the $k$-th column of the matrix $2\omega$.

\begin{proposition}
Using~\eqref{polara1} the formula~\eqref{main} can be rewritten in the form suitable for the inversion of the integral of the third kind
\begin{align}\begin{split}
&\int_{P'}^P \sum_{k=1}^g \left[ \frac{y+W_k}{x-Z_k} -\frac{y+W_k'}{x-Z_k'}    \right]\frac{\mathrm{d}x}{2y}\\&=-
\int_{P'}^P \mathrm{d}\boldsymbol{u}^T(x,y) \int_{\mathcal{D}'}^{\mathcal{D}} \mathrm{d}\boldsymbol{r}(z,w)+\mathrm{ln} \frac{\sigma\left(\int_{\infty}^P \mathrm{d}\boldsymbol{u}- \int_{g\infty}^{\mathcal{D}} \mathrm{d}\boldsymbol{u}  \right)}{\sigma\left(\int_{\infty}^P \mathrm{d}\boldsymbol{u}- \int_{g\infty}^{\mathcal{D}'} \mathrm{d}\boldsymbol{u}  \right)}
-\mathrm{ln} \frac{\sigma\left(\int_{\infty}^{P'} \mathrm{d}\boldsymbol{u}- \int_{g\infty}^{\mathcal{D}} \mathrm{d}\boldsymbol{u}  \right)}{\sigma\left(\int_{\infty}^{P'} \mathrm{d}\boldsymbol{u}- \int_{g\infty}^{\mathcal{D}'} \mathrm{d}\boldsymbol{u}  \right)} \ , \end{split}  \label{main1}
\end{align}
where the expression $\int_{\mathcal{D}'}^{\mathcal{D}} \mathrm{d}\boldsymbol{r}(z,w)$ can be calculated by~\eqref{JIPmeromorphic}.
\end{proposition}

\subsection{Stratification of the $\theta$-divisor and inversion}

The $\theta$-divisor $\widetilde{\Theta}$ is defined as the subset of $\widetilde{\mathrm{Jac}}(X_g)$ that nullifies the $\theta$-function and, therefore, the $\sigma$-function, i.e.
\begin{equation}
\widetilde{\Theta}=\left\{ \boldsymbol{v}\in \widetilde{\mathrm{Jac}}(X_g) \, \vert \,\theta(\boldsymbol{v})\equiv 0 \right\} \,.
\label{Theta1}
\end{equation}
The subset $\widetilde{\Theta}_k\subset \widetilde{\Theta}$, $0 \leq k < g$, is called $k$-th stratum if each point $\boldsymbol{v}\in \widetilde{\Theta}$ admits a parametrization
\begin{equation}
\widetilde{\Theta}_k := \Biggl\{\boldsymbol{v} \in \widetilde{\Theta} \Big|
\boldsymbol{v} = \sum_{j=1}^{k}\int_{\infty}^{P_j} \mathrm{d}\boldsymbol{v} + \boldsymbol{K}_{\infty}\Biggr\} \, , \label{strata}
\end{equation}
where $\widetilde{\Theta}_0=\{\boldsymbol{K}_{\infty}\}$ and $\widetilde{\Theta}_{g-1} = \widetilde{\Theta}$. We furthermore denote $\widetilde{\Theta}_g= \widetilde{\mathrm{Jac}}(X_g)$ and we have the natural embedding
\begin{equation}
\label{stratification} \widetilde{\Theta}_0\;\subset\; \widetilde{\Theta}_1\; \subset \ldots \subset\; \widetilde{\Theta}_{g-1}\; \subset \; \widetilde{\Theta}_g = \widetilde{\mathrm{Jac}}(X_g) \,.
\end{equation}

We define the $\theta$-function to be {\it vanishing to the order} $m(\widetilde{\Theta}_k)$ along the stratum $\widetilde{\Theta}_k$ if for all sets $\alpha_j$, $j = 1, \ldots, g$ with $0 \leq \alpha_1 + \ldots + \alpha_g < m$ holds
\begin{equation}
\frac{\partial^{\alpha_1+\ldots+\alpha_g}}{\partial u_1^{\alpha_1} \ldots \partial u_g^{\alpha_g} }\theta(\boldsymbol{v\vert\tau})\equiv 0,\quad \forall \boldsymbol{v}\in \widetilde{\Theta}_k \, , \label{vanish}
\end{equation}
and there is a certain set of $\alpha_j$, with $\alpha_1 + \ldots + \alpha_g = m$ such that (\ref{vanish}) does not hold. The orders $m(\widetilde{\Theta}_k)$ of the vanishing of $\theta(\widetilde{\Theta}_k + \boldsymbol{v})$ along the stratum $\widetilde{\Theta}_k$ for some genera are given in Table~\ref{table1}.

In the following we focus on the stratum $\widetilde{\Theta}_1$ corresponding to the variety $\Theta_1 \subset \mathrm{Jac}(X_g)$, which is the image of the curve inside the Jacobian,
\begin{equation}
\widetilde{\Theta}_1 := \Biggl\{\boldsymbol{v} \in \widetilde{\Theta} \Big|
\boldsymbol{v} = \int_{\infty}^{P} \mathrm{d}\boldsymbol{v} + \boldsymbol{K}_{\infty}\Biggr\} \, .
\end{equation}

We remark that another stratification was introduced in \cite{vanha95} for hyperelliptic curves of even order with two infinite points $\infty_+$ and $\infty_-$ that was implemented for studying the poles of functions on Jacobians of these curves. The same problem relevant to strata of the $\theta$-divisor was studied in~\cite{abendfed00}.

\begin{table}[t]
\begin{tabular}{ |l|l|l|l|l|l|l|l|l| }
$g$&$m(\widetilde{\Theta}_0)$&$m(\widetilde{\Theta}_1)$&$m(\widetilde{\Theta}_2)$&
$m(\widetilde{\Theta}_3)$&$m(\widetilde{\Theta}_4)$&$m(\widetilde{\Theta}_5)$&$m(\widetilde{\Theta}_6)$\\
\hline
1&1&0&-&-&-&-&-\\
2&1&1&0&-&-&-&-\\
3&2&1&1&0&-&-&-\\
4&2&2&1&1&0&-&-\\
5&3&2&2&1&1&0&-\\
6&3&3&2&2&1&1&0
\end{tabular}
\vskip0.3cm
\caption{Orders $m(\widetilde{\Theta}_k)$ of zeros $\theta(\widetilde{\Theta}_k+\boldsymbol{v})$ at $\boldsymbol{v}=0$ on the strata $\widetilde{\Theta}_k$.} \label{table1}
\end{table}

\section{Inversion of hyperelliptic integrals of higher genera}\label{sec:invhige}

\subsection{Inversion of holomorphic integrals}

For the case of genus two the inversion of a holomorphic hyperelliptic integral by the method of restriction to the $\theta$-divisor was obtained independently by Grant \cite{gr90} and Jorgenson \cite{jo92} in the form
\begin{equation}
x=-\left.\frac{\sigma_1(\boldsymbol{u})}{\sigma_2(\boldsymbol{u})}\right|_{\sigma(\boldsymbol{u})=0}, \qquad \boldsymbol{u}=(u_1,u_2)^T\label{gr} \,.
\end{equation}
This result was implemented in \cite{EPR03}, and explicitly worked out in the series of publications \cite{HackmannLaemmerzahl08,HackmannLaemmerzahl08a,Hackmannetal08,Hackmannetal09,Hackmannetal010}, and others.

The case of genus three was studied by \^Onishi \cite{on98}, where the inversion formula is given in the form
\begin{equation}
x=-\left.\frac{\sigma_{13}(\boldsymbol{u})}{\sigma_{23}(\boldsymbol{u})}\right|_{\sigma(\boldsymbol{u})=\sigma_3(\boldsymbol{u})=0}, \qquad \boldsymbol{u}=(u_1,u_2,u_3)^T\label{onishi} \,.
\end{equation}
Formula (\ref{onishi}) is based on the detailed analysis of the genus three KdV hierarchy and its restriction to the $\theta$-divisor. Below we will present the generalization of (\ref{gr}) and (\ref{onishi}) to higher genera. For doing this we first analyze the Schur--Weierstra{\ss} polynomials that represent the first term of the expansion of $\sigma(\boldsymbol{u})$ in the vicinity of the origin $\boldsymbol{u}\sim 0$. The $\theta$-divisor $\Theta$ and its strata $\Theta_k$ in the vicinity of the origin $\boldsymbol{u}\sim 0$ are given as polynomials in $\boldsymbol{u}$.

An analysis of the Schur--Weierstra{\ss} polynomials leads to 
\begin{proposition}\label{SWPOL}
The following statements are valid for the Schur--Weierstra{\ss} polynomials $S_{\boldsymbol{\pi}}(\boldsymbol{u})$
associated with a partition $\boldsymbol{\pi}$:
\begin{enumerate}

\item 
In the vicinity of the origin, an element $\boldsymbol{u}$ of the first stratum $\Theta_1 \subset \Theta$ is singled out by
\begin{equation}
S_{\boldsymbol{\pi}}(\boldsymbol{u})=0,\quad \frac{\partial^j }{\partial u_g^j}  S_{\boldsymbol{\pi}}(\boldsymbol{u})=0 \quad \forall\,j=1,\ldots,g-2 \,.
\end{equation}

\item The derivatives fulfil
\begin{equation}
\frac{\partial^j }{\partial u_g^j}  S_{\boldsymbol{\pi}}(\boldsymbol{u})
\begin{cases}
\equiv 0 \quad \text{if} \quad   1\leq j < \frac{g(g-1)}{2}\\
\not\equiv 0 \quad \text{if} \quad   j \geq \frac{g(g-1)}{2}
\end{cases} \text{with}\quad \boldsymbol{u}\in {\Theta}_1 \, .
\end{equation}

\item The following equalities are valid for $\boldsymbol{u}\in {\Theta}_1$
\begin{equation}
x\cong-\frac{1}{u_g^2}=- \frac{
 \dfrac{\partial^{M}}{\partial u_1 \partial u_{g}^{M-1} }  S_{\boldsymbol{\pi}}(\boldsymbol{u})}
{\dfrac{\partial^{M}}{\partial u_2\partial u_{g}^{M-1} } S_{\boldsymbol{\pi}}(\boldsymbol{u})}=-
\frac{\dfrac{\partial^{M+1}}{\partial u_1 \partial u_{g}^M }  S_{\boldsymbol{\pi}}(\boldsymbol{u})}
{\dfrac{\partial^{M+1}}{\partial u_2\partial u_{g}^M } S_{\boldsymbol{\pi}}(\boldsymbol{u})}  \, ,
\end{equation}
where $M= \frac12 (g-2)(g-3) + 1$.
\item The order of vanishing of $S_{\boldsymbol{\pi}}$ restricted to $\Theta_1$ is the rank of the partition $\pi$.
\end{enumerate}
\end{proposition}

It was noted in \cite{BEL99} that the Schur--Weierstra{\ss} polynomials respect all statements of the Riemann singularity theorem. In particular, if
\begin{equation}
\boldsymbol{Z}=\left(\frac{z^{2g-1}}{2g-1},\ldots, \frac{z^{2k-1}}{2k-1}\ldots,\frac{z^3}{3},z \right)
\end{equation}
and if $\boldsymbol{\pi}$ is the partition at the infinite Weierstra{\ss} point of the hyperelliptic curve $X_g$ of genus $g$, then the function
\begin{equation}
G(z) := S_{\boldsymbol{\pi}}(\boldsymbol{Z}-\boldsymbol{u})
\end{equation}
either has $g$ zeros or vanishes identically.
Moreover we will conjecture here that the properties of the Schur--Weierstra{\ss} polynomials given in Proposition \ref{SWPOL} can be ``lifted" to the fundamental $\sigma$-function~\eqref{sigma}.

The above analysis permits to conjecture the following inversion formula for the general case of hyperelliptic curves of genus $g>2$~\cite{ehkkl11}
\begin{equation}
x=-\left. \frac{
 \dfrac{\partial^{M+1}}{\partial u_1 \partial u_{g}^M }  \sigma(\boldsymbol{u})}
{\dfrac{\partial^{M+1}}{\partial u_2\partial u_{g}^M } \sigma(\boldsymbol{u})} \right|_{\boldsymbol{u}\in \Theta_1}, \qquad M=\frac{(g-2)(g-3)}{2}+1
\label{matprev}
\end{equation}
and
\begin{equation}
{\Theta}_1 = \left\{\boldsymbol{u} \in {\rm Jac}(X_g)\; \Big|\; \sigma(\boldsymbol{u})=0,\; \frac{\partial^j }{\partial u_g^j} \sigma(\boldsymbol{u})=0 \quad \forall\,j=1,\ldots, g-2\right\} \,. \label{restriction}
\end{equation}

The analog of this formula for strata $\Theta_{k}$, $ 1<k<g $ and $(n,s)$-curves in the terminology of \cite{BEL99} was recently considered by Matsutani and Previato   \cite{matprev08}, \cite{matprev10}.

\subsection{Inversion of the integrals of the second kind }

Also the meromorphic integrals can be expressed in terms of the $\sigma$-functions restricted to the stratum $\Theta_1$ of the $\theta$-divisor. To do that we consider the $\zeta$-formula~\eqref{JIPmeromorphic} and move to infinity the $g-1$ points of the divisor $\mathcal{D}$ on both sides of the equality. The poles on both sides cancel and the non-vanishing terms yield the required formula. We demonstrate this procedure here for the example of the genus three case given in \cite{eg03} (the analysis of the genus two case can be found in \cite{EPR03}). For that we use the following proposition

\begin{proposition} Let $X_3$ be the genus three curve $y^2=4x^{7}+\lambda_6 x^6+\ldots+\lambda_0$.
Let $\boldsymbol{u}\in \Theta_1=\{\boldsymbol{u}\vert \sigma(\boldsymbol{u})=\sigma_3(\boldsymbol{u})=0\}$. Then the following formulae are valid
\begin{align}\begin{split}
&\int_{P_0}^P\mathrm{d}r_3=-\frac{\sigma_{23}(\boldsymbol{u})}{\sigma_2(\boldsymbol{u})}+c_3 \ ,\\
&\int_{P_0}^P\mathrm{d}r_2=-\frac12\frac{\sigma_{22}(\boldsymbol{u})}{\sigma_2(\boldsymbol{u})}+c_2 \ , \\
&\int_{P_0}^P\mathrm{d}r_1=\frac12\frac{\sigma_1(\boldsymbol{u})\sigma_{22}(\boldsymbol{u})}{\sigma_2(\boldsymbol{u})^2}-\frac{\sigma_{12}(\boldsymbol{u})}{\sigma_2(\boldsymbol{u})}+c_1 \ , \end{split}  \label{gibbons}
\end{align}
where $P=(x,y)$ is an arbitrary point of the curve, and $P_0=(x_0,y_0)\neq(\infty,\infty)$ is any fixed point. It is convenient to choose in particular $P_0=(e_{2g},0)$. The constants $c_i$ are fixed by the requirement that the right hand side vanishes at $P=P_0$
\end{proposition}

\begin{proof}
In the case of $g=3$ the vector $\boldsymbol{\mathfrak{Z}}(\boldsymbol{u})$~\eqref{zwp} yields
\begin{equation}  \boldsymbol{\mathfrak{Z}}(\boldsymbol{u}) \left( \begin{array}{c} \mathfrak{Z}_1(\boldsymbol{u})\\
 \mathfrak{Z}_2(\boldsymbol{u})\\
 \mathfrak{Z}_3(\boldsymbol{u})
 \end{array}  \right)  = \left( \begin{array}{c}  \wp_{33}(\boldsymbol{u}) \wp_{333}(\boldsymbol{u})+2\wp_{233}(\boldsymbol{u})\\
 \wp_{333}(\boldsymbol{u})\\0
 \end{array} \right) \ . \end{equation}
First we restrict relation~\eqref{JIPmeromorphic} to the stratum $\Theta_2$. To do that we use the following expansions
\begin{align}
\left.\int_{\infty}^{P_3}\mathrm{d}\boldsymbol{u}\right|_{x_3=1/\xi^2}=\left(\begin{array}{c}-\frac15\xi^5+{\mathscr{O}}(\xi^7)\\
-\frac13\xi^3+{\mathscr{O}}(\xi^5)\\-\xi+{\mathscr{O}}(\xi^3)
 \end{array}\right), \quad \left.\int_{(e_6,0)}^{P_3}\mathrm{d}\boldsymbol{r}\right|_{x_3=1/\xi^2}=\left(\begin{array}{c}\xi^{-5}+{\mathscr{O}}(\xi^{-3})\\
\xi^{-3}+{\mathscr{O}}(\xi^{-1})\\
\xi^{-1}+{\mathscr{O}}(\xi)
 \end{array}\right)
\end{align}
taking into account the condition $\sigma(\boldsymbol{u})=0$ in the expansions. Then we restrict the resulting formulae to the stratum $\Theta_1$. To do that we expand the holomorphic and meromorphic integrals  $\int_{\infty}^{P_2}\mathrm{d}\boldsymbol{u}$ and    $\int_{(e_4,0)}^{P_2}\mathrm{d}\boldsymbol{r}$ in the vicinity of $P_2=(\infty,\infty)$ taking into account the condition $\sigma_3(\boldsymbol{u})=0$. From that the relations~\eqref{gibbons} follow. For the simplification of the obtained relations we used the formulae from \cite{on98} below, which result from the restriction to $\Theta_1$ of the KdV hierarchy associated with the genus three curve
\begin{eqnarray}
\left.\left\{\sigma_{333}(\boldsymbol{u})-2\sigma_2(\boldsymbol{u})\right\}\right|_{\Theta_1} &=& 0 \nonumber \\
\left.\left\{\sigma_{233}(\boldsymbol{u})-\frac{\sigma_{23}(\boldsymbol{u})^2}{\sigma_2(\boldsymbol{u})}+\sigma_1(\boldsymbol{u})\right\}\right|_{\Theta_1} &=& 0\\
\left.\left\{\sigma_{133}(\boldsymbol{u})-\frac{\sigma_1(\boldsymbol{u})\sigma_{23}(\boldsymbol{u})^2}{\sigma_2(\boldsymbol{u})^2}-\frac{\sigma_1(\boldsymbol{u})^2}{\sigma_2(\boldsymbol{u})}\right\}\right|_{\Theta_1}&=&0 \nonumber \ .
\end{eqnarray}
\end{proof}
Cases of higher genera cases can be considered in an analogously way though it is not possible to present a general formula.

\subsection{Inversion of the integral of the third kind}\label{subsec:main1}

We consider~\eqref{main1} when all points $P_k$, $P_k'$ have the same coordinate $Z_k'=Z_k$ but $W_k'=-W_k$ and choose the divisors
\begin{equation}
\mathcal{D}= \{(Z,W), (e_4,0),\ldots, (e_{2g},0)   \},\quad \mathcal{D}'= \{ (Z,W'), (e_4,0),\ldots, (e_{2g},0)\} \ . 
\end{equation}
It is evident that
\begin{align}
\int_{ \mathcal{D}' }^{\mathcal{D}} \mathrm{d}\boldsymbol{r}(z,w)=\int_{ (e_2,0) }^{(Z,W)} \mathrm{d}\boldsymbol{r}(z,w)-\int_{ (e_2,0) }^{(Z,W')} \mathrm{d}\boldsymbol{r}(z,w) \ .
\end{align}
From the $\zeta$--formula (\ref{JIPmeromorphic}) we get
\begin{align*}
\int_{ (e_2,0) }^{(Z,W)} \mathrm{d}\boldsymbol{r}(z,w)=-\boldsymbol{\zeta}\left( \int_{(e_2,0)}^{(Z,W)} \mathrm{d}\boldsymbol{u} + \boldsymbol{K}_\infty \right) + 2( \boldsymbol{\eta}^{\prime}\boldsymbol{\varepsilon}^\prime + \boldsymbol{\eta}\boldsymbol{\varepsilon} )+\frac12\boldsymbol{\mathfrak{Z}}(Z,W) ,\\
 \int_{ (e_2,0) }^{(Z,-W)} \mathrm{d}\boldsymbol{r}(z,w)=\boldsymbol{\zeta}\left( \int_{(e_2,0)}^{(Z,W)} \mathrm{d}\boldsymbol{u} + \boldsymbol{K}_\infty \right) - 2( \boldsymbol{\eta}^{\prime}\boldsymbol{\varepsilon}^\prime + \boldsymbol{\eta}\boldsymbol{\varepsilon} )
 - \frac12\boldsymbol{\mathfrak{Z}}(Z,W),
\end{align*}
where $\mathfrak{Z}_g(Z,W)=0$ and for $1\leq j<g$ we have
\begin{equation}
\mathfrak{Z}_j(Z,W)=\frac{W}{\prod_{k=2}^{g} (Z-e_{2k})}\sum_{k=0}^{g-j-1}(-1)^{g-k+j+1}Z^kS_{g-k-j-1}(\boldsymbol{e}) \, .
 \end{equation}
The $S_k(\boldsymbol{e})$ are elementary symmetric functions of order $k$ built on  $g-1$ branch points $e_4,\ldots, e_{2g}$: $S_0=1$, $S_1=e_4+\ldots+e_{2g}$ etc. Then the solution of the inversion problem for the integral of the third kind~\eqref{main1} takes the form
\begin{align}\begin{split}
W\int_{P'}^P\frac{1}{x-Z}\frac{\mathrm{d}x}{y} = & 2\int_{P'}^P \mathrm{d}\boldsymbol{u}^T(x,y) \left[ \boldsymbol{\zeta} \left( \int_{(e_2,0)}^{(Z,W)} \mathrm{d} \boldsymbol{u} + \boldsymbol{K}_\infty  \right) - 2( \boldsymbol{\eta}^{\prime}\boldsymbol{\varepsilon}^\prime + \boldsymbol{\eta}\boldsymbol{\varepsilon} )  - \frac12 \boldsymbol{\mathfrak{Z}}(Z,W)    \right]\\
& +\ln \frac{\sigma\left(\int_{\infty}^P \mathrm{d}\boldsymbol{u}- \int_{(e_2,0)}^{(Z,W)} \mathrm{d}\boldsymbol{u} - \boldsymbol{K}_\infty  \right)}{\sigma\left(\int_{\infty}^P \mathrm{d}\boldsymbol{u}+ \int_{(e_2,0)}^{(Z,W)} \mathrm{d}\boldsymbol{u} - \boldsymbol{K}_\infty \right)}
- \mathrm{ln} \frac{\sigma\left(\int_{\infty}^{P'} \mathrm{d}\boldsymbol{u} - \int_{(e_2,0)}^{(Z,W)} \mathrm{d}\boldsymbol{u} - \boldsymbol{K}_\infty  \right)}{\sigma\left(\int_{\infty}^{P'} \mathrm{d}\boldsymbol{u} + \int_{(e_2,0)}^{(Z,W)} \mathrm{d}\boldsymbol{u} - \boldsymbol{K}_\infty  \right)}.\end{split} \label{main1-2}
\end{align}

\noindent Note, that~\eqref{main1-2} represents a natural generalization of the known antiderivative for elliptic functions
\begin{equation}
\wp'(v)\int\frac{\mathrm{d}u}{\wp(u)- \wp(v)}=2u\zeta(v)+\ln\sigma(u-v)-\ln\sigma(u+v) \ . \label{maingen1}
\end{equation}

The inversion procedure for the integral of the third kind 
\begin{equation}
 W\int_{(e_2,0)}^P\frac{1}{x-Z}\frac{\mathrm{d}z}{w}=t, \label{thirdkindinversion}
\end{equation}
is described as follows. Here for simplicity the lower bound is fixed at a branch point, say at $P'=(e_2,0)$.

\noindent In this case the formula~\eqref{main1-2} takes the form
\begin{eqnarray}
&& W\int_{(e_2,0)}^P\frac{1}{x-Z}\frac{\mathrm{d}x}{y}= 2(\boldsymbol{\mathfrak{A}}_2-\boldsymbol{u}) \left[ \boldsymbol{\zeta} (\boldsymbol{v} + \boldsymbol{K}_\infty ) - 2( \boldsymbol{\eta}^{\prime}\boldsymbol{\varepsilon}^\prime + \boldsymbol{\eta}\boldsymbol{\varepsilon} )  - \frac12 \boldsymbol{\mathfrak{Z}}(Z,W) \right] \nonumber \\
&&
\qquad\qquad+\mathrm{ln} \frac{\sigma\left(\boldsymbol{u}- \boldsymbol{v} - \boldsymbol{K}_\infty \right)}
{\sigma\left(\boldsymbol{u}+ \boldsymbol{v} - \boldsymbol{K}_\infty \right)}
-\mathrm{ln} \frac{\sigma\left(\boldsymbol{\mathfrak{A}}_2- \boldsymbol{v} - \boldsymbol{K}_\infty  \right)}{\sigma\left(\boldsymbol{\mathfrak{A}}_2+ \boldsymbol{v} - \boldsymbol{K}_\infty \right)} \ ,  \label{main2}
\end{eqnarray}
where 
\begin{equation}
 \boldsymbol{v}=\int_{{(e_2,0)}}^{(Z,W)}\mathrm{d}\boldsymbol{u} \ , \quad   \boldsymbol{\mathfrak{A}}_2=\int_{\infty}^{(e_2,0)}\mathrm{d}\boldsymbol{u} \, \quad\text{and} \, \quad \boldsymbol{u}\in \Theta_1, \quad \boldsymbol{u}=\int_{\infty}^{(x,y)}\mathrm{d}\boldsymbol{u} \nonumber \ .
\end{equation}

\subsection{Solving the inversion problem}

Now we are in a position to describe the inversion procedure for the integral~\eqref{inversion1}. The formulae given above permit to represent equation~\eqref{inversion1} in the form
\begin{equation}
\mathcal{F}( x, u_1,\ldots,u_g, \mathrm{constants}  )=t \ .
\end{equation}
The function $\mathcal{F}$ depends on $x$ via the $g$ variables $u_1,\ldots, u_g$ that are abelian images of $(x,y)$ and various constants given by the poles of the integrals of the third kind and the coefficients $a_k,b_k,c_k$, and possibly via the elementary function $\mathcal{E}(x)$ in the equation~\eqref{inversion2}.  This relation is complemented by the $g-1$ conditions $\boldsymbol{u}\in\Theta_1$
\begin{equation}
\sigma(\boldsymbol{u})=0,\quad \frac{\partial^j }{\partial u_g^j} \sigma(\boldsymbol{u})=0 \quad \forall\,j=1,\ldots, g-2 \ . \label{thirdcond1-2}
\end{equation}
From these relations one can find (numerically) the functions $u_1=u_1(t),\ldots,u_g=u_g(t)$ and plug these into~\eqref{matprev} to find $x=x(t)$.

\section{ Computer algebra supporting the method}~\label{sec:coalsume}
Presently effective means of computer algebra are developed to execute the above claimed program of integral inversion. We will consider for this, e.g., the Maple/algcurves code.

\subsection{Riemann period matrix and winding vectors}
For a given curve of genus $g$ we compute first the period matrices $(2\omega,2\omega')$ and $\tau=\omega^{-1}\omega'$ by means of the Maple/algcurves code. From that we determine the winding vectors, i.e. the columns of the inverse matrix,
\begin{equation}
(2\omega)^{-1}=(\boldsymbol{U}_1,\ldots, \boldsymbol{U}_g) \ . \label{windvec}
\end{equation}

\subsection{Homology basis} In our analysis we used a specific homology basis for a hyperelliptic curve (see Fig.\ref{figure-1}). That was done just to clarify the approach.
But the result should be independent of the choice of the homology basis. It is possible to perform all the calculations without making an explicit plot of the homology basis and to use that one given in the Tretkoff-Tretkoff construction that is programmed in the Maple/algcurves. Nevertheless, in certain cases we should know the correspondence between the branch points and half-periods in the homology basis used by Maple/algcurves. In particular the proposed method of inversion supposes the knowledge of this correspondence. One can find this correspondence using the generalized Bolza formulae. That can be done as follows.

We first find all nonsingular odd characteristics by direct computation of all odd $\theta$-constants. According to Table \ref{table1} we have two sets $B_1 \subset \widetilde{\Theta}_{g-1}$ and $B_2 \subset \widetilde{\Theta}_{g-2}$ of nonsingular odd half--periods. For each element of $b_1 \in B_1$ there are $e_{i_1}, \ldots, e_{i_{g-1}} \neq \infty$ such that
\begin{equation}
b_1 = \int_{\infty}^{(e_{i_1},0)} \mathrm{d}\boldsymbol{v}+\ldots+\int_{\infty}^{(e_{i_{g-1}},0)} \mathrm{d}\boldsymbol{v}+\boldsymbol{K}_{\infty} \in \widetilde{\Theta}_{g-1}
\end{equation}
and for each element of $b_2 \in B_2$ there are $e_{i_1}, \ldots, e_{i_{g-2}} \neq \infty$ such that
\begin{equation}
b_2 = \int_{\infty}^{(e_{i_1},0)} \mathrm{d}\boldsymbol{v}+\ldots+\int_{\infty}^{(e_{i_{g-2}},0)} \mathrm{d}\boldsymbol{v}+\boldsymbol{K}_{\infty} \in \widetilde{\Theta}_{g-2} \, .
\end{equation}

\noindent
Using the known values of the winding vectors (\cite{ehkkl11}, Proposition 4.3) one can find the correspondence between the sets $\{e_{i_1},\ldots, e_{i_{g-1}} \}$ and $\{e_{i_1},\ldots, e_{i_{g-2}} \}$ of branch points and the nonsingular odd characteristics $[ (2\omega)^{-1}\left(\boldsymbol{\mathfrak A}_{i_1,\ldots,i_{g-1}} \right) + \boldsymbol{K}_{\infty} ]$ and $[ (2\omega)^{-1}\left(\boldsymbol{\mathfrak A}_{i_1,\ldots,i_{g-2}} \right) + \boldsymbol{K}_{\infty}]$ . Then one can add these characteristics and find the one--to--one correspondence
\begin{equation}
\int_{\infty}^{(e_{i_{g-1}},0)}  \mathrm{d}\boldsymbol{v}  \leftrightarrows \;[\boldsymbol{\mathfrak A}_{i_{g-1}}] ,\quad i=1,\ldots, 2g+2 \,. \label{correspondence}
\end{equation}

\subsection{Second period matrix}
We present the formula that permits to express the second period matices $2\eta,2\eta'$ and $\varkappa$ in terms of the first period matrices $2\omega,2\omega'$, the branch points and the $\theta$-constants
\begin{proposition} \label{kappaprop}
Let $\boldsymbol{\mathfrak A}_{\mathcal{I}_0}+2\omega \boldsymbol{K}_{\infty}$ be an arbitrary even nonsingular half-period corresponding to the $g$ branch points of the set of indices $\mathcal{I}_0=\{i_1,\ldots,i_g\}$.
We define the symmetric $g\times g$ matrices
\begin{equation}
\mathfrak{P}(\boldsymbol{\mathfrak A}_{\mathcal{I}_0}) := \left(\wp_{ij}(\boldsymbol{\mathfrak A}_{\mathcal{I}_0})   \right)_{i,j=1,\ldots,g}
\label{matrixP}
\end{equation}
and the $g\times g$ matrix $H$ which is expressible in terms of the even non-singular $\theta$-constants
\footnote{In the previous work~\cite{ehkkl11} this formula numbered as (3.50) was written in an equivalent but more complicated form with misplaced sign ``-". }
\begin{equation}
H(\boldsymbol{\mathfrak A}_{\mathcal{I}_0}) =  \frac{1}{\theta[\varepsilon]} \left( \theta_{ij}[\varepsilon]
 \right)_{i,j=1,\ldots,g}, \quad\text{where}\quad \quad \varepsilon=[(2\omega)^{-1}\boldsymbol{\mathfrak A}_{\mathcal{I}_0}+\boldsymbol{K}_{\infty}] \ .
\end{equation}
Then the $\varkappa$-matrix is given by
\begin{equation}
\varkappa = - \frac12\mathfrak{P}(\boldsymbol{\mathfrak A}_{\mathcal{I}_0})-\frac12
((2\omega)^{-1})^T
H(\boldsymbol{\mathfrak A}_{\mathcal{I}_0}) (2\omega)^{-1} \ . \label{kappa}
\end{equation}
The half-periods $\eta$ and $\eta'$ of the meromorphic differentials can be represented as
\begin{equation}
\eta = 2 \varkappa \omega, \qquad \eta' = 2 \varkappa \omega' - \frac{\imath\pi}{2}(\omega^{-1})^T \,.
\end{equation}
\end{proposition}

We remark that~\eqref{kappa} represents the natural generalization of the Weierstra{\ss} formulae\footnote{The correspondence between the numeration of branch points $e_i$ and $\theta$-constants $\vartheta_j(0)$ depends on the chosen homology basis}
\begin{equation}
2\eta\omega=-2e_1\omega^2-\frac12 \frac{\vartheta_2''(0)}{\vartheta_2(0)}\,,\quad
2\eta\omega=-2e_2\omega^2-\frac12 \frac{\vartheta_3''(0)}{\vartheta_3(0)}\,,\quad
2\eta\omega=-2e_3\omega^2-\frac12 \frac{\vartheta_4''(0)}{\vartheta_4(0)}\, , \label{weierstrassformulae}
\end{equation}
see e.g. the Weierstra{\ss}--Schwarz lectures, \cite{weierstrass893} p. 44. From~\eqref{weierstrassformulae} follows
\begin{equation}
\omega\eta=-\frac{1}{12}\left(\frac{\vartheta''_2(0)}{\vartheta_2(0)}
+\frac{\vartheta''_3(0)}{\vartheta_3(0)}+\frac{\vartheta''_4(0)}{\vartheta_4(0)}  \right) \ .
\label{omegaeta}
\end{equation} 
Therefore Proposition~\ref{kappaprop} allows the reduction of the variety of moduli necessary for the calculation of the $\sigma$- and $\wp$-functions to the first period matrix. The generalization of \eqref{omegaeta} to arbitrary higher genera algebraic curves has recently been  discussed in~\cite{ksh11} where its role in the construction of invariant generalizations of the Weierstra{\ss} $\sigma$--function to higher genera is elucidated.

For the following Lemma we introduce $g$ vectors
\[ \boldsymbol{E}_j=(1,e_j,\ldots,e_j^{g-1})^T,\quad j\in \mathcal{I}_0 \ .\]

\begin{lemma}~\label{lemma:p-matrix} Let $S_k^{(ij)}$ be elementary symmetric functions of order $k$ built on the $2g-1$ elements $\{e_1,\ldots, e_{2g+1}\} / \{e_i, e_j\}$. Then the following relations are valid
\begin{align}
\boldsymbol{E}_i^T \mathfrak{P}(\boldsymbol{\mathfrak A}_{\mathcal{I}_0}) \boldsymbol{E}_j\equiv\frac{F(e_i,e_j)}{4(e_i-e_j)^2}
\equiv\sum_{n=0}^{g-1} e_i^ne_j^nS^{(ij)}_{2g-2n-1} \ . \label{relation}
\end{align}
Here $F(x,z)$ is the Kleinian 2-polar~\eqref{2polar}.
\end{lemma}

\begin{proof}
The LHS of the relation~\eqref{relation} follows from~\eqref{kleinformula}
\begin{equation}
\sum_{i,j=1}^g\wp_{ij}\left(  \sum_{k=1}^g \int_{P_0}^{P_k} \mathrm{d}\boldsymbol{u}   \right)x_r^{i-1}x_s^{j-1} =\frac{F(x_r,x_s)-2y_ry_s }{4(x_r-x_s)^2}, \quad r\neq s, \quad P_k=(x_k,y_k) \ , \label{LHS_relE}
\end{equation}
where the argument of the $\wp$-function is a non-singular even half period given as an abelian image of $g$ branch points $P_k=(e_{i_k},0)$, $i_1,\ldots,i_g\in \{1,\ldots, 2g + 1\}$. The branch points $e_r$ and $e_s$ on the right hand side of~\eqref{LHS_relE} belong to the chosen subset of $g$ branch points. The RHS of the relation~\eqref{relation} can be checked directly for small genera and presents a conjecture for higher genera.

\end{proof}

Varying the integers $i,j$ along the set $\mathcal{I}_0$ one can obtain from (IV.10) $g(g-1)/2$ equations with respect $g(g-1)/2$ components, $\wp_{g-1,g-1},\ldots, \wp_{11} $ of the matrix 
$\mathfrak{P}(\boldsymbol{\mathfrak{A}}_{\mathcal{I}_0})$ and then solve them by the Kramer rule. Indeed, since the components $\wp_{gi}(\mathcal{I}_0)$ are already known from the solution of the Jacobi inversion problem (II.58) then Eq. (IV.10) can be simplified for every pair $i\neq j \in \mathcal{I}_{p}$ as follows

\begin{equation}
\widehat{\boldsymbol{E}}^T_i \widehat{\mathfrak{P}}(\boldsymbol{\mathfrak{A}}_{\mathcal{I}_0})\widehat{\boldsymbol{E}}_j=
e_i^{g-1}e_j^{g-1}\left(S_1^{(ij)}+ \sum_{k\in \mathcal{I}_0/\{i,j\}}e_k\right)
+ \sum_{n=0}^{g-2}e_i^{n}e_j^n S^{(ij)}_{2g-2n-1}
\end{equation}
where 

\begin{equation}
\widehat{\mathfrak{P}}(\boldsymbol{\mathfrak{A}}_{\mathcal{I}_0})=( \wp_{ij}(\boldsymbol{\mathfrak{A}}_{\mathcal{I}_0}) )_{i,j=1,\ldots,g-1}, \quad \widehat{\boldsymbol{E}}_i=(1,e_i,\ldots,e_i^{g-2})^T \ .
\end{equation}

Examples for genus two, three and four in Sections~\ref{sec:curvesgenus2}-\ref{sec:curvesgenus4} show that the entries in $\widehat{ \mathfrak{P}}$ are polynomials in the branch points $e_i$ (see Eqns.~\eqref{JIPE2},~\eqref{JIPE32} and~\eqref{JIPE4_wp}). We infer that this is true in general but do not present here the general expression of these polynomials.

\subsection{Characteristics} 

We explained above how to find the correspondence between the half-periods and the abelian images of the branch points on the basis of Bolza type formulae. This correspondence is necessary for the description of the real evolution of the system that corresponds to the motion of the divisor points, i.e. functions on the upper bound in the Abel map, over the Riemann surface between certain branch points or between the corresponding half-periods in the Jacobi variety. We intend to reduce a number of cases when complete hyperelliptic integrals are calculated. Finding the above correspondence allows to present initial/final points of the evolution in terms of a half period.

\subsection{Vector of Riemann constants} 

The vector of Riemann constants $\boldsymbol{K}_{\infty}$ enters the derived inversion formulae and can be computed as follows. It was proved that the vector of Riemann constants with the base point fixed in a branch point, e.g. infinity in our considerations, is a half-period. From Table~\ref{table1} for the stratum $\widetilde{\Theta}_0$ follows whether $\boldsymbol{K}_{\infty}$ is even or odd (parity of $m$) and whether it is singular ($m>1$). In accordance with the definition of the $\theta$-divisor it is sufficient to find the even or odd half-period $\boldsymbol{K}_{\infty}$ which satisfies the condition
\[   \theta\left(   \sum_{k=1}^{g-1} \int_{\infty}^{P_k}\mathrm{d} \boldsymbol{v} +\boldsymbol{K}_{\infty}  \right) =0,  \]
where $P_1,\ldots,P_{g-1}$ are $g-1$ arbitrary points on the curve $X_g$.

Alternatively one can use the correspondence between the branch points and the half periods. Among the $2g+2$ characteristics  (\ref{correspondence}) there should be precisely $g$ odd and $g+2$ even characteristics. The sum of all odd characteristics gives the vector of Riemann constants with the base point at infinity. This characteristic will then be singular of  order $ \left[\frac{g+1}{2}\right]$.

\section{Hyperelliptic curve of genus two}\label{sec:curvesgenus2}

We consider a hyperelliptic curve $X_2$ of genus two
\begin{align}
\begin{split}
w^2 & = 4(z-e_1)(z-e_2)(z-e_3)(z-e_4)(z-e_5)\\
& = 4 z^5 + \lambda_4 z^4 + \lambda_3 z^3 + \lambda_2 z^2 + \lambda_1 z + \lambda_0 \, .
\end{split} \label{curve2}
\end{align}
From~\eqref{holomorphicdiff} and~\eqref{meromorphicdiff} the basic holomorphic and meromorphic differentials are
\begin{align}
\mathrm{d}u_1 &=\frac{\mathrm{d}z}{w}\,, & \qquad \mathrm{d}r_1 & = \frac{12 z^3+2\lambda_4z^2+\lambda_3z}{4w}\mathrm{d}z \,, \\
\mathrm{d}u_2 &=\frac{z\mathrm{d}z}{w}\,, & \qquad \mathrm{d}r_2 & = \frac{z^2}{w}\mathrm{d}z \,.
\end{align}
Then the Jacobi inversion problem for the equations
\begin{align}\begin{split}
\int_{\infty}^{(z_1,w_1)}\frac{ \mathrm{d}z}{w}+\int_{\infty}^{(z_2,w_2)}\frac{ \mathrm{d}z}{w}=u_1 \,,\\
\int_{\infty}^{(z_1,w_1)}\frac{z \mathrm{d}z}{w}+\int_{\infty}^{(z_2,w_2)}\frac{z \mathrm{d}z}{w}=u_2\end{split} \label{JIP}
\end{align}
is solved according to~\eqref{JIP1} and~\eqref{JIP2} in the form
\begin{align}\begin{split}
z_1+z_2&=\wp_{22}(\boldsymbol{u}), \quad z_1z_2=-\wp_{12}(\boldsymbol{u}) \,,\\
w_k&= \wp_{222}(\boldsymbol{u})z_k + \wp_{122}(\boldsymbol{u}), \quad k=1,2 \,.
\end{split} \label{SOLJIP}
\end{align}

\subsection{Characteristics in genus two}
The homology basis of the curve is fixed by defining the set of
half-periods corresponding to the branch points. The characteristics of the abelian images of the branch points are defined as
\begin{equation}
[\boldsymbol{\mathfrak A}_i]=\left[\int_{\infty}^{(e_i,0)} \mathrm{d}\boldsymbol{u}\right] = \begin{pmatrix} \boldsymbol{\varepsilon}_i^{'^T} \\ \boldsymbol{\varepsilon}_i^{^T} \end{pmatrix} = \begin{pmatrix} \varepsilon_{i,1}' & \varepsilon_{i,2}' \\ \varepsilon_{i,1} & \varepsilon_{i,2} \end{pmatrix} \,,
\end{equation}
which can be also written as
\[\boldsymbol{\mathfrak A}_i=2\omega  \boldsymbol{\varepsilon}_i+ 2\omega' \boldsymbol{\varepsilon'}_i, \quad i=1,\ldots,6 \,. \]

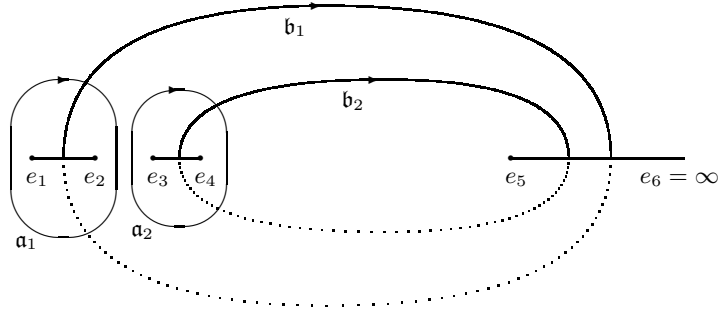
\begin{figure}
\begin{center}
\unitlength 0.7mm \linethickness{0.6pt}
\begin{picture}(150.00,80.00)
\put(9.,33.){\line(1,0){12.}} \put(9.,33.){\circle*{1}}
\put(21.,33.){\circle*{1}} \put(10.,29.){\makebox(0,0)[cc]{$e_1$}}
\put(21.,29.){\makebox(0,0)[cc]{$e_2$}}
\put(15.,33.){\oval(20,30.)}
\put(8.,17.){\makebox(0,0)[cc]{$\mathfrak{ a}_1$}}
\put(15.,48.){\vector(1,0){1.0}}
\put(32.,33.){\line(1,0){9.}} \put(32.,33.){\circle*{1}}
\put(41.,33.){\circle*{1}} \put(33.,29.){\makebox(0,0)[cc]{$e_3$}}
\put(42.,29.){\makebox(0,0)[cc]{$e_4$}}
\put(37.,33.){\oval(18.,26.)}
\put(30.,19.){\makebox(0,0)[cc]{$\mathfrak{a}_2$}}
\put(36.,46.){\vector(1,0){1.0}}
\put(100.,33.00) {\line(1,0){33.}} \put(100.,33.){\circle*{1}}
\put(101.,29.){\makebox(0,0)[cc]{$e_{5}$}}
\put(132.,29.){\makebox(0,0)[cc]{$e_{6}=\infty$}}
\put(59.,58.){\makebox(0,0)[cc]{$\mathfrak{b}_1$}}
\put(63.,62.){\vector(1,0){1.0}}
\bezier{484}(15.,33.00)(15.,62.)(65.,62.)
\bezier{816}(65.00,62.)(119.00,62.00)(119.00,33.00)
\bezier{35}(15.,33.00)(15.,5.)(65.,5.)
\bezier{35}(65.00,5.)(119.00,5.00)(119.00,33.00)
\put(70.,44.){\makebox(0,0)[cc]{$\mathfrak{b}_2$}}
\put(74.00,48.){\vector(1,0){1.0}}
\bezier{384}(37.,33.00)(37.,48.)(76.00,48.)
\bezier{516}(76.00,48.)(111.00,48.00)(111.00,33.00)
\bezier{30}(37.,33.00)(37.,19.)(76.00,19.)
\bezier{30}(76.00,19.)(111.00,19.00)(111.00,33.00)
\end{picture}
\end{center}
\caption{Homology basis on the Riemann surface of the curve $X_2$ with real branch points $e_1 < e_2 <\ldots < e_{6}=\infty$ (upper sheet). The cuts are drawn from $e_{2i-1}$ to $e_{2i}$, $i=1,2,3$. The $\mathfrak{b}$--cycles are completed on the lower sheet (dotted lines).} \label{figure-2}
\end{figure}

In the homology basis given in Figure~\ref{figure-2} the characteristics of the branch points are
\begin{align}
[\boldsymbol{\mathfrak A}_1] & = \frac{1}{2}
\begin{pmatrix} 1 & 0 \\ 0 & 0 \end{pmatrix} \, , & \quad
[\boldsymbol{\mathfrak A}_2] & = \frac{1}{2}
\begin{pmatrix} 1 & 0 \\ 1 & 0 \end{pmatrix} \, , & \quad
[\boldsymbol{\mathfrak A}_3] & = \frac{1}{2}
\begin{pmatrix} 0 & 1 \\ 1 & 0 \end{pmatrix} \\
[\boldsymbol{\mathfrak A}_4] & = \frac{1}{2}
\begin{pmatrix} 0 & 1 \\ 1 & 1 \end{pmatrix}\, , & \quad
[\boldsymbol{\mathfrak A}_5] & = \frac{1}{2}
\begin{pmatrix} 0 & 0 \\ 1 & 1 \end{pmatrix}\, , & \quad
[\boldsymbol{\mathfrak A}_6] & = \frac{1}{2}
\begin{pmatrix} 0 & 0 \\ 0 & 0 \end{pmatrix} \, .
\label{hombas2}
\end{align}
The characteristics of the vector of Riemann constants $\boldsymbol{K}_{\infty}$ yield
\begin{equation}
[\boldsymbol{K}_{\infty}] = [\boldsymbol{\mathfrak{A}}_2]+ [\boldsymbol{\mathfrak{A}}_4] = \frac{1}{2} \begin{pmatrix} 1 & 1 \\ 0 & 1 \end{pmatrix} \, .
\end{equation}

From the above characteristics we build 16 half-periods. Denote 10 half-periods for $i\neq j = 1,\ldots, 5$ that are images of two branch points as
\begin{align}
\boldsymbol{\Omega}_{ij}=2\omega ( \boldsymbol{\varepsilon}_i+\boldsymbol{\varepsilon}_j)+
2\omega' (\boldsymbol{\varepsilon}_i'+  \boldsymbol{\varepsilon}_j') ,\quad  i = 1, \ldots, 5 \,.
\label{even}
\end{align}
Then the characteristics of the $6$ half-periods
\begin{equation}
\left[ (2\omega)^{-1} \boldsymbol{\mathfrak A}_i+\boldsymbol{K}_{\infty} \right] =: \delta_i  ,\quad i= 1, \ldots, 6
\end{equation}
are nonsingular and odd, whereas the characteristics of the $10$ half-periods
\begin{equation}
\left[(2\omega)^{-1} \boldsymbol{\Omega}_{ij}+\boldsymbol{K}_{\infty}\right] =: \varepsilon_{ij}, \quad 1 \leq i < j \leq 5
\end{equation}
are nonsingular and even.

Odd characteristics correspond to partitions $\{6\}\cup \{ 1,\ldots,5 \}$ and
$\{k\}\cup\{i_1,\ldots,i_4,6\}$ for $i_1,\ldots,i_4\neq k$. The first partition from these two corresponds to $\Theta_0$ and the second to $\Theta_1$.

From the solution of the Jacobi inversion problem we obtain for any $i,j =1,\ldots,5$, $i\neq j$
\begin{equation}
e_i+e_j=\wp_{22}(\boldsymbol{\Omega}_{ij}),\quad -e_ie_j=\wp_{12}(\boldsymbol{\Omega}_{ij})\,.
\label{JIPE1}
\end{equation}
From the relation
\[ \wp_{11}(\boldsymbol{u}) = \frac{F(x_1,x_2)-2y_1y_2}{4(x_1-x_2)^2} \]
one can also find
\begin{equation}
e_ie_j(e_p+e_q+e_r)+e_pe_qe_r=\wp_{11}(\boldsymbol{\Omega}_{ij}) \,,
\label{JIPE2}
\end{equation}
where $i$,$j$,$p$,$q$, and $r$ are mutually different.

From~\eqref{JIPE1} and~\eqref{JIPE2} we obtain an expression for the matrix $\varkappa$~\eqref{kappa} that is useful for numeric calculations because it reduces the second period matrix to an expression in the first period matrix and $\theta$-constants, namely, in the case $e_i=e_1,e_j=e_2$,
\begin{equation}
\varkappa=-\frac12 \left(\begin{array}{cc} e_1e_2(e_3+e_4+e_5)+e_3e_4e_5&-e_1e_2\\
-e_1e_2&e_1+e_2 \end{array}\right)-\frac12 {(2\omega)^{-1}}^T \frac{1}{\theta[\varepsilon]} \left( \begin{array}{cc}
\theta_{11}[\varepsilon]&\theta_{12}[\varepsilon]\\
\theta_{12}[\varepsilon]&\theta_{22}[\varepsilon] \end{array} \right)(2\omega)^{-1} \,,
\label{kappa2}
\end{equation}
where the characteristic $\varepsilon$ in the fixed homology basis reads
\[  \varepsilon = [\boldsymbol{\mathfrak{A}}_1]+ [\boldsymbol{\mathfrak{A}}_2]+[\boldsymbol{K}_{\infty}]
=\left(\begin{array}{cc} \frac12& \frac12\\
\frac12&\frac12\end{array}\right).   \]

\subsection{Inversion of a holomorphic integral}

Taking the limit $z_2\rightarrow \infty$ in the Jacobi inversion problem (\ref{JIP}) we obtain
\begin{equation}
\int_{\infty}^{(z,w)} \frac{\mathrm{d}z}{w}=u_1,\quad \int_{\infty}^{(z,w)} \frac{z\mathrm{d}z}{w}=u_2 \,.
\end{equation}
The same limit in the ratio
\begin{equation}
\frac{\wp_{12}(\boldsymbol{u})}{\wp_{22}(\boldsymbol{u})} = - \frac{z_1z_2}{z_1+z_2}
\end{equation}
leads to the Grant-Jorgenson formula (\ref{gr}).
In terms of $\theta$-functions this can be given the form
\begin{equation}
z = \left.-\frac{\partial_{\boldsymbol{U}_1}  \theta[\boldsymbol{K}_{\infty}]((2\omega)^{-1}\boldsymbol{u};\tau)}{\partial_{\boldsymbol{U}_2}  \theta[\boldsymbol{K}_{\infty}]((2\omega)^{-1}\boldsymbol{u};\tau)}\right|_{\theta((2\omega)^{-1}\boldsymbol{u};\tau)=0},
\end{equation}
where here and below $\partial_{\boldsymbol{U}_1} = \sum_{j=1}^g {U_1}_j \frac{\partial}{\partial z_j}$ is the derivative along the direction $\boldsymbol{U}_1$. The ``winding vectors" $\boldsymbol{U}_1$, $\boldsymbol{U}_2$ are the column vectors of the inverse matrix $(2\omega)^{-1}$~\eqref{windvec}.

From~\eqref{gr} we obtain for all finite branch points
\begin{equation}
e_i=-\frac{\sigma_1(\boldsymbol{\mathfrak A}_i ) }{\sigma_2(\boldsymbol{\mathfrak A}_i) }, \quad \text{or equivalently}, \quad e_i=-\frac{\partial_{\boldsymbol{U}_1}\theta[\delta_i]}{\partial_{\boldsymbol{U}_2}\theta[\delta_i]}, \quad i=1,\ldots,5 \label{e2a} \,.
\end{equation}
This formula was mentioned by Bolza \cite{bolza86} (see his Eq.~(6)) for the case of a genus two curve with finite branch points.

The $\zeta$-formula reads

\begin{eqnarray}
-\zeta_1(\boldsymbol{u})+2\mathfrak{n}_1+\frac12 \frac{w_1-w_2}{z_1-z_2}&=&\int_{(e_2,0)}^{(z_1,w_1)}\mathrm{d}r_1(z,w)+
\int_{(e_4,0)}^{(z_2,w_2)}\mathrm{d}r_1(z,w) \ , \nonumber \\
-\zeta_2(\boldsymbol{u})+2\mathfrak{n}_2&=&\int_{(e_2,0)}^{(z_1,w_1)}\mathrm{d}r_2(z,w)+
\int_{(e_4,0)}^{(z_2,w_2)}\mathrm{d}r_2(z,w) \ ,  \label{zetaformula2}
\end{eqnarray}
where $\mathfrak{n}_j=\sum_{i=1}^{2}\eta^{\prime}_{ji}{\varepsilon}'_{i} + \eta_{ji}{\varepsilon}_{i}$. Here the characteristics ${\varepsilon}'_{i}$ and ${\varepsilon}_{i}$ of $\boldsymbol{K}_{\infty}$ are not reduced. Choosing $(z_1,w_1)=(Z,W)$, $(z_2,w_2)=(e_4,0) $ we get from~\eqref{zetaformula2}

\begin{eqnarray}
-\zeta_1\left(\int_{(e_2,0)}^{(Z,W)} \mathrm{d}\boldsymbol{u} + \boldsymbol{K}_{\infty} \right)+2\mathfrak{n}_1 +\frac12\frac{W}{Z-e_4}&=&\int_{(e_2,0)}^{(Z,W)}\mathrm{d}r_1(z,w) \ , \nonumber \\
-\zeta_2\left(\int_{(e_2,0)}^{(Z,W)} \mathrm{d}\boldsymbol{u} + \boldsymbol{K}_{\infty} \right)+2\mathfrak{n}_2&=&\int_{(e_2,0)}^{(Z,W)}\mathrm{d}r_2(z,w) \ .
 \label{zeta2}
\end{eqnarray}

The inversion formula for the integral of the third kind~\eqref{main1-2} is written as

 \begin{align}\begin{split}
& W\int_{P'}^P\frac{1}{x-Z}\frac{\mathrm{d}x}{y}
=-2\left(\boldsymbol{u}^T-{\boldsymbol{u}'}^T\right)
\int_{(e_2,0)}^{(Z,W)} \mathrm{d}\boldsymbol{r}
+\mathrm{ln} \frac{\sigma\left(\boldsymbol{u}-\boldsymbol{v} - \boldsymbol{K}_{\infty} \right)}{\sigma\left(\boldsymbol{u}+\boldsymbol{v} - \boldsymbol{K}_{\infty}  \right)}-\mathrm{ln} \frac{\sigma\left(\boldsymbol{u}'-\boldsymbol{v} - \boldsymbol{K}_{\infty} \right)}{\sigma\left(\boldsymbol{u}'+\boldsymbol{v} - \boldsymbol{K}_{\infty}  \right)}
\end{split} \label{thirdinv2}
\end{align}
with
\[\boldsymbol{v}=\int_{(e_2,0)}^{(Z,W)}\mathrm{d}\boldsymbol{u},\quad \boldsymbol{u}=\int_{\infty}^{P}\mathrm{d}\boldsymbol{u},\quad \boldsymbol{u}'=\int_{\infty}^{P'}\mathrm{d}\boldsymbol{u} \]
and $\boldsymbol{u}\in\Theta_1$, $\boldsymbol{u}'\in\Theta_1$. The integrals $\displaystyle{\int_{(e_2,0)}^{(Z,W)} \mathrm{d}\boldsymbol{r}}$ are given by the formula~\eqref{zeta2}.

In the case when the base point $P'$ is chosen to be a branch point, say $(e_2,0)$ then the final formula takes the form

\begin{eqnarray}
 W\int_{(e_2,0)}^{P}\frac{1}{x-Z}\frac{\mathrm{d}x}{y}
&=& 2\left(\boldsymbol{u}^T-\boldsymbol{\mathfrak{A}}_2^T\right) \left[\boldsymbol{\zeta}(\boldsymbol{v}+\boldsymbol{K}_{\infty})-2( \eta^{\prime}\boldsymbol{\varepsilon}'_{\boldsymbol{K}_{\infty}} + \eta\boldsymbol{\varepsilon}_{\boldsymbol{K}_{\infty}} )-\frac12 \boldsymbol{\mathfrak{Z}}(Z,W)\right]   \nonumber \\ & + &\mathrm{ln} \frac{\sigma\left(\boldsymbol{u}-\boldsymbol{v} - \boldsymbol{K}_{\infty} \right)}{\sigma\left(\boldsymbol{u}+\boldsymbol{v} - \boldsymbol{K}_{\infty}  \right)}-\mathrm{ln} \frac{\sigma\left(\boldsymbol{\mathfrak{A}}_2-\boldsymbol{v} - \boldsymbol{K}_{\infty} \right)}{\sigma\left(\boldsymbol{\mathfrak{A}}_2+\boldsymbol{v}  - \boldsymbol{K}_{\infty} \right)} \ .
 \label{thirdinv2final1}
\end{eqnarray}


\section{Hyperelliptic curve of genus three}\label{sec:curvesgenus3}

We consider also a hyperelliptic curve $X_3$ of genus three with seven real zeros given by
\begin{align}\begin{split}
w^2&=4(z-e_1)(z-e_2)(z-e_3)(z-e_4)(z-e_5)(z-e_6)(z-e_7)\\
&=4z^{7}+\lambda_6z^6+\ldots+\lambda_1z+\lambda_0 \,.
\end{split}
\label{curve3}
\end{align}

The complete set of holomorphic and meromorphic differentials with a unique pole at infinity is
\begin{align}
\mathrm{d}u_1 & = \frac{\mathrm{d}z}{w}\,, & \qquad \mathrm{d}r_1 & = z(20z^4+4\lambda_6z^3+3\lambda_5z^2+2\lambda_4z+\lambda_3)
 \frac{\mathrm{d}z}{4w} \,, \nonumber \\
\mathrm{d}u_2 &= \frac{z\mathrm{d}z}{w}\,, & \qquad \mathrm{d}r_2 & = z^2(12z^2+2\lambda_6z+\lambda_5)\frac{\mathrm{d}z}{4w}\,, \\
\mathrm{d}u_3 &= \frac{z^2\mathrm{d}z}{w}\,, & \qquad  \mathrm{d}r_3 & = \frac{z^3\mathrm{d}z}{w} \,. \nonumber
\end{align}

The Jacobi inversion problem for the equations
\begin{align}\begin{split}
\int_{\infty}^{(z_1,w_1)}\frac{ \mathrm{d}z}{w}
+\int_{\infty}^{(z_2,w_2)}\frac{ \mathrm{d}z}{w}
+\int_{\infty}^{(z_3,w_3)}\frac{ \mathrm{d}z}{w}
=u_1,\\
\int_{\infty}^{(z_1,w_1)}\frac{z \mathrm{d}z}{w}
+\int_{\infty}^{(z_2,w_2)}\frac{z \mathrm{d}z}{w}
+\int_{\infty}^{(z_3,w_3)}\frac{z \mathrm{d}z}{w}
=u_2,\\
\int_{\infty}^{(z_1,w_1)}\frac{z^2 \mathrm{d}z}{w}
+\int_{\infty}^{(z_2,w_2)}\frac{z^2 \mathrm{d}z}{w}
+\int_{\infty}^{(z_3,w_3)}\frac{z^2 \mathrm{d}z}{w}
=u_3\end{split} \label{JIP3}
\end{align}
is solved by
\begin{align}\begin{split}
z_1+z_2+z_3&=\wp_{33}(\boldsymbol{u}), \quad z_1z_2+z_1z_3+z_2z_3=-\wp_{23}(\boldsymbol{u}), \quad z_1z_2z_3= \wp_{13}(\boldsymbol{u})\\
w_k&=\wp_{333}(\boldsymbol{u})z_k^2+\wp_{233}(\boldsymbol{u})z_k + \wp_{133}(\boldsymbol{u}) , \quad k=1,2,3 \,.
\end{split} \label{SOLJIP3}
\end{align}

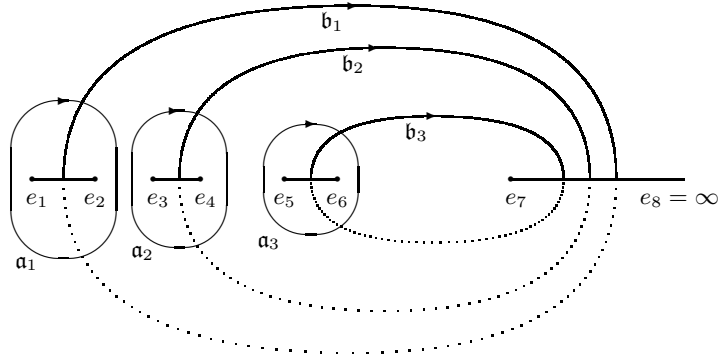
\begin{figure}
\begin{center}
\unitlength 0.7mm \linethickness{0.6pt}
\begin{picture}(150.00,80.00)
\put(9.,33.){\line(1,0){12.}} \put(9.,33.){\circle*{1}}
\put(21.,33.){\circle*{1}} \put(10.,29.){\makebox(0,0)[cc]{$e_1$}}
\put(21.,29.){\makebox(0,0)[cc]{$e_2$}}
\put(15.,33.){\oval(20,30.)}
\put(8.,17.){\makebox(0,0)[cc]{$\mathfrak{ a}_1$}}
\put(15.,48.){\vector(1,0){1.0}}
\put(32.,33.){\line(1,0){9.}} \put(32.,33.){\circle*{1}}
\put(41.,33.){\circle*{1}} \put(33.,29.){\makebox(0,0)[cc]{$e_3$}}
\put(42.,29.){\makebox(0,0)[cc]{$e_4$}}
\put(37.,33.){\oval(18.,26.)}
\put(30.,19.){\makebox(0,0)[cc]{$\mathfrak{a}_2$}}
\put(36.,46.){\vector(1,0){1.0}}
\put(57.,33.){\line(1,0){10.}} \put(57.,33.){\circle*{1}}
\put(67.,33.){\circle*{1}} \put(57.,29.){\makebox(0,0)[cc]{$e_5$}}
\put(67.,29.){\makebox(0,0)[cc]{$e_6$}}
\put(62.,33.){\oval(18.,21.)}
\put(54.,21.){\makebox(0,0)[cc]{$\mathfrak{a}_3$}}
\put(62.,43.5){\vector(1,0){1.0}}
\put(100.,33.00) {\line(1,0){33.}} \put(100.,33.){\circle*{1}}
\put(101.,29.){\makebox(0,0)[cc]{$e_{7}$}}
\put(132.,29.){\makebox(0,0)[cc]{$e_{8}=\infty$}}
\put(66.,63.){\makebox(0,0)[cc]{$\mathfrak{b}_1$}}
\put(70.,66.){\vector(1,0){1.0}}
\bezier{484}(15.,33.)(15.,66.)(70.,66.)
\bezier{816}(70.,66.)(120.,66.00)(120.,33.)
\bezier{35}(15.,33.)(15.,0.)(70.,0.)
\bezier{35}(70.,0.)(120.,0.)(120.,33.)
\put(70.,55.){\makebox(0,0)[cc]{$\mathfrak{b}_2$}}
\put(74.00,58.){\vector(1,0){1.0}}
\bezier{384}(37.,33.)(37.,58.)(76.,58.)
\bezier{516}(76.,58.)(115.,58.)(115.00,33.00)
\bezier{30}(37.,33.00)(37.,8.)(76.00,8.)
\bezier{30}(76.00,8.)(115.00,8.00)(115.00,33.00)
\put(82.,42.){\makebox(0,0)[cc]{$\mathfrak{b}_3$}}
\put(85,45.){\vector(1,0){1.0}}
\bezier{384}(62.,33.)(62.,45.)(85.,45.)
\bezier{516}(85.,45.)(110.,45.)(110.00,33.00)
\bezier{30}(62.,33.00)(62.,21.)(85.00,21.)
\bezier{30}(85.00,21.)(110.00,21.00)(110.00,33.00)
\end{picture}
\end{center}
\caption{Homology basis on the Riemann surface of the curve
$X_3$ with real branch points $e_1 < e_2 <\ldots <
e_{8}=\infty$ (upper sheet).  The cuts are drawn from $e_{2i-1}$
to $e_{2i}$, $i=1,2,4$.  The $\mathfrak{b}$--cycles are completed on the
lower sheet (dotted lines).} \label{figure-3}
\end{figure}

\subsection{Characteristics in genus three}

Let $\mathfrak{A}_k$ be the abelian image of the $k$-th branch point, namely
\begin{equation}
\boldsymbol{\mathfrak{A}}_k=\int_{\infty}^{(e_k,0)} \mathrm{d}\boldsymbol{u}= 2\omega \boldsymbol{\varepsilon}_k+2\omega' \boldsymbol{\varepsilon}_k', \quad k=1,\ldots,8 \,,
\end{equation}
where $\boldsymbol{\varepsilon}_k$ and $\boldsymbol{\varepsilon}_k'$ are column vectors whose entries $\varepsilon_{k,j}$, $\varepsilon'_{k,j}$ are $\frac{1}{2}$ or $0$ for all $k=1,\ldots,8$, $j=1,2,3$.

The correspondence between the branch points and the characteristics in the fixed homology basis is given as
\begin{align}\begin{split}
[\boldsymbol{{\mathfrak A}}_1]= \frac{1}{2}
\begin{pmatrix} 1 & 0 & 0 \\ 0 & 0 & 0 \end{pmatrix}\, ,\quad
[\boldsymbol{{\mathfrak A}}_2]= \frac{1}{2}
\begin{pmatrix} 1 & 0 & 0 \\ 1 & 0 & 0 \end{pmatrix}\, ,\quad
[\boldsymbol{{\mathfrak A}}_3] = \frac{1}{2}
\begin{pmatrix} 0 & 1 & 0 \\ 1 & 0 & 0 \end{pmatrix}\, , \\
[\boldsymbol{{\mathfrak A}}_4]= \frac{1}{2}
\begin{pmatrix} 0 & 1 & 0 \\ 1 & 1 & 0 \end{pmatrix}\, ,\quad
[\boldsymbol{{\mathfrak A}}_5] = \frac{1}{2}
\begin{pmatrix} 0 & 0 & 1 \\ 1 & 1 & 0 \end{pmatrix}\, ,\quad
[\boldsymbol{{\mathfrak A}}_6]= \frac{1}{2}
\begin{pmatrix} 0 & 0 & 1 \\ 1 & 1 & 1 \end{pmatrix}\, , \\
[\boldsymbol{{\mathfrak A}}_7]= \frac{1}{2}
\begin{pmatrix} 0 & 0 & 0 \\ 1 & 1 & 1 \end{pmatrix}\, ,\quad
[\boldsymbol{{\mathfrak A}}_8]= \frac{1}{2}
\begin{pmatrix} 0 & 0 & 0 \\ 0 & 0 & 0 \end{pmatrix}\, . \end{split}
\label{hombasis_gen3}
\end{align}

The vector of Riemann constants $\boldsymbol{K}_{\infty}$ with the base point at infinity is given in the above basis by the even singular characteristics,
\begin{equation}
[\boldsymbol{K}_{\infty}]=[\boldsymbol{\mathfrak A}_2]+[\boldsymbol{\mathfrak A}_4]+[\boldsymbol{\mathfrak A}_6] = \frac12 \begin{pmatrix} 1 & 1 & 1 \\ 1 & 0 & 1 \end{pmatrix}  \, .
\end{equation}

From the above characteristics the 64 half-periods can be built as follows. If we start with singular even characteristics, then there should be only one such characteristic that corresponds to the vector of Riemann constants $\boldsymbol{K}_{\infty}$. The corresponding partition reads $\mathcal{I}_2\cup \mathcal{J}_2 = \{ \} \cup \{1, 2, \ldots, 8\}$ and the $\theta$-function $\theta(\boldsymbol{K}_{\infty}+\boldsymbol{v})$ vanishes at the origin $\boldsymbol{v}=0$ to the order $m=2$.

The half-periods $\boldsymbol{\Delta}_1=(2\omega)^{-1}\boldsymbol{\mathfrak A}_k+\boldsymbol{K}_{\infty}\in \Theta_1$ correspond to partitions
\begin{equation}
\mathcal{I}_1\cup \mathcal{J}_1 =  \{ k,8 \} \cup \{ j_1,\ldots, j_6 \}, \quad j_1,\ldots,j_6 \notin \{8,k\}
\end{equation}
and the $\theta$-function $\theta(\boldsymbol{\Delta}_1+\boldsymbol{v})$ vanishes at the origin $\boldsymbol{v}=0$ to the order $m=1$.

We also denote the $21$ half-periods that are images of two branch points
\begin{align}
\boldsymbol{\Omega}_{ij}=2\omega ( \boldsymbol{\varepsilon}_i+\boldsymbol{\varepsilon}_j)+2\omega' (\boldsymbol{\varepsilon}_i'+  \boldsymbol{\varepsilon}_j') ,\quad  i,j = 1,\ldots, 7, i\neq j \ .
\label{odd2}
\end{align}
The half-periods $\boldsymbol{\Delta}_1=(2\omega)^{-1}\boldsymbol{\Omega}_{ij}
+\boldsymbol{K}_{\infty}\in \Theta_2$ correspond to the partitions
\begin{equation}
\mathcal{I}_1\cup \mathcal{J}_1 = \{ i,j \} \cup \{ j_1,\ldots, j_6 \}, \quad j_1,\ldots,j_6 \notin \{i,j\}
\end{equation}
and the $\theta$-function $\theta(\boldsymbol{\Delta}_1 + \boldsymbol{v})$ vanishes at the origin, $\boldsymbol{v}=0$, as before to the order $m=1$.
Therefore the characteristics of the $7$ half-periods
\begin{equation}
\left[(2\omega)^{-1} \boldsymbol{\mathfrak A}_i+\boldsymbol{K}_{\infty}\right] =: \delta_i \, ,\quad i=1,\ldots,7
\end{equation}
are nonsingular and odd as well as the characteristics of the $21$ half-periods
\begin{equation}
\left[(2\omega)^{-1} \boldsymbol{\Omega}_{ij}+\boldsymbol{K}_{\infty}\right] =: \varepsilon_{ij}   \, ,\quad 1\leq i<j\leq 7 \,.
\end{equation}
We finally introduce the $35$ half-periods that are images of three branch points
\begin{align}
\boldsymbol{\Omega}_{ijk}=2\omega ( \boldsymbol{\varepsilon}_i+\boldsymbol{\varepsilon}_j+\boldsymbol{\varepsilon}_k)+ 2\omega' (\boldsymbol{\varepsilon}_i'+  \boldsymbol{\varepsilon}_j'+  \boldsymbol{\varepsilon}_k')\in \mathrm{Jac}(X_g) ,\quad  1\leq i<j<k\leq 7 \,.
\end{align}
The half-periods $\boldsymbol{\Delta}_2=(2\omega)^{-1}\boldsymbol{\Omega}_{ijk}
+\boldsymbol{K}_{\infty}$ correspond to the partitions
\begin{equation}
\mathcal{I}_0\cup \mathcal{J}_0 =  \{ i,j,k,8 \} \cup \{ j_1,\ldots, j_4 \}, \quad j_1,\ldots,j_4 \notin \{i,j,k,8\}\,.
\end{equation}
The $\theta$-function $\theta(\boldsymbol{\Delta}_2+\boldsymbol{v})$
does not vanish at the origin $\boldsymbol{v}=0$.

Furthermore, the $35$ characteristics
\begin{equation}
\varepsilon_{ijk}=  \left[ (2\omega)^{-1} \boldsymbol{\Omega}_{ijk}+\boldsymbol{K}_{\infty} \right],\quad 1\leq i<j<k\leq 7
\end{equation}
are even and nonsingular while the characteristic $[\boldsymbol{K}_{\infty}]$ is even and singular. Altogether we got all $64=4^3$ characteristics classified by the partitions of the branch points.

\subsection{Inversion of a holomorphic integral}

All three holomorphic integrals,
\begin{align}
\int_{\infty}^{(x,w)} \frac{\mathrm{d}z}{w}=u_1,\quad
\int_{\infty}^{(x,w)} \frac{z\mathrm{d}z}{w}=u_2,\quad
\int_{\infty}^{(x,w)} \frac{z^2\mathrm{d}z}{w}=u_3
\end{align}
are inverted by the same formula (\ref{onishi}). Nevertheless, there are three different cases for which one of the variables $u_1,u_2,u_3$ is considered as independent while the remaining two result from solving the divisor conditions $\sigma(\boldsymbol{u})=\sigma_3(\boldsymbol{u})=0$.

Formula (\ref{onishi}) can be rewritten in terms of $\theta$-functions as
\begin{align}
x=-\frac{\partial^2_{\boldsymbol{U}_1,\boldsymbol{U}_3}
\theta[\boldsymbol{K}_{\infty}]((2\omega)^{-1}\boldsymbol{u})
+2(\partial_{\boldsymbol{U}_1}
\theta[\boldsymbol{K}_{\infty}]((2\omega)^{-1}\boldsymbol{u}))
\boldsymbol{e}_3^T\varkappa\boldsymbol{u}}{\partial^2_{\boldsymbol{U}_2,\boldsymbol{U}_3}
\theta[\boldsymbol{K}_{\infty}]((2\omega)^{-1}\boldsymbol{u})
+2(\partial_{\boldsymbol{U}_2}
\theta[\boldsymbol{K}_{\infty}]((2\omega)^{-1}\boldsymbol{u}))
\boldsymbol{e}_3^T\varkappa\boldsymbol{u}}, \label{xtheta}
\end{align}
where $\boldsymbol{e}_3=(0,0,1)^T$. This represents the solution of the inversion problem.

From the solution of the Jacobi inversion problem follows for any $1\leq i<j<k \leq 7$,
\begin{equation}
e_i+e_j+e_k=\wp_{33}(\boldsymbol{\Omega}_{ijk}),\quad
-e_ie_j-e_ie_k-e_je_k=\wp_{23}(\boldsymbol{\Omega}_{ijk}),\quad
e_ie_je_k=\wp_{13}(\boldsymbol{\Omega}_{ijk}) \,.
\label{JIPE31}
\end{equation}
Solving equations~\eqref{relation} we find
\begin{equation}
\wp_{12}(\boldsymbol{\Omega}_{ijk})=-s_3S_1-S_4 \, , \quad \wp_{11}(\boldsymbol{\Omega}_{ijk})=s_3S_2+s_1S_4 \, , \quad \wp_{22}(\boldsymbol{\Omega}_{ijk})=S_3+2s_3+s_2S_1 \,,
\label{JIPE32}
\end{equation}
where the $s_l$ are the elementary symmetric functions of order $l$ of the branch points $e_i,e_j,e_k$ and $S_l$ are the elementary symmetric functions of order $l$ of the remaining branch points $\{1,\ldots,7\} \setminus \{i,j,k\}$.

From~\eqref{kappa} using~\eqref{JIPE31} and~\eqref{JIPE32} one can find the expression for the matrix $\varkappa$. To do that we take the half-period $\boldsymbol{\Omega}_{123}$
\begin{align}
\varkappa = -\frac12 \mathfrak{P}(\boldsymbol{\Omega}_{123}) -\frac12 {(2\omega)^{-1}}^T H(\boldsymbol{\Omega}_{123}) (2\omega)^{-1}
\end{align}
with
\begin{equation}
H(\boldsymbol{\Omega}_{123})=\frac{1}{\theta[\varepsilon]}\left(\begin{array}{ccc}
\theta_{11}[\varepsilon]&\theta_{12}[\varepsilon]&\theta_{13}[\varepsilon]\\
\theta_{12}[\varepsilon]&\theta_{22}[\varepsilon]&\theta_{23}[\varepsilon]\\
\theta_{13}[\varepsilon]&\theta_{23}[\varepsilon]&\theta_{33}[\varepsilon]
 \end{array}\right),\quad \varepsilon=\left( \begin{array}{ccc} \frac12&0& \frac12\\  \frac12&0& \frac12  \end{array}   \right) \ .
\end{equation}

For the branch points $e_1,\ldots, e_8$ the expression
\begin{equation}
e_i=-\frac{\partial_{\boldsymbol{U}_1}\left[
\partial_{\boldsymbol{U}_3}+2\boldsymbol{\mathfrak A}^T_i\varkappa \boldsymbol{e}_3
\right]\theta[\boldsymbol{K}_{\infty}]((2\omega)^{-1} \boldsymbol{\mathfrak A}_i;\tau)}
{\partial_{\boldsymbol{U}_2}\left[
\partial_{\boldsymbol{U}_3}+2\boldsymbol{\mathfrak A}^T_i\varkappa \boldsymbol{e}_3
\right]\theta[\boldsymbol{K}_{\infty}]((2\omega)^{-1} \boldsymbol{\mathfrak A}_i;\tau)}
\label{thomae1}
\end{equation}
is valid. Furthermore we have for $i,j=1,\ldots,8$, $i\neq j$
\begin{align}\begin{split}
e_i + e_j & = - \frac{\sigma_2(\boldsymbol{\Omega}_{ij})}
{\sigma_3(\boldsymbol{\Omega}_{ij})} \equiv
\frac{\partial_{\boldsymbol{U}_2} \theta[\varepsilon_{ij}]}
{\partial_{\boldsymbol{U}_3} \theta[\varepsilon_{ij}]}, \\
e_i e_j & = \frac{\sigma_1(\boldsymbol{\Omega}_{ij})}
{\sigma_3(\boldsymbol{\Omega}_{ij})} \equiv
\frac{\partial_{\boldsymbol{U}_2} \theta[\varepsilon_{ij}]}
{\partial_{\boldsymbol{U}_3} \theta[\varepsilon_{ij}]} \, ,
\end{split}\label{thomae2}
\end{align}
and for $i=1,\ldots,7$
\begin{equation}
e_i = -\frac{\sigma_1(\boldsymbol{\mathfrak A}_i)} {\sigma_2(\boldsymbol{\mathfrak A}_i)}
= - \frac{\partial_{\boldsymbol{U}_1} \theta[\delta_{i}]}
{\partial_{\boldsymbol{U}_2} \theta[\delta_{i}]} \, .
\end{equation}

The $\zeta$-formula reads

\begin{align}\begin{split}
&-\zeta_1(\boldsymbol{u})+2\mathfrak{n}_1+\frac12 \frac{w_1(z_1-z_2-z_3)}{(z_1-z_2)(z_1-z_3)}+\text{permutations}\\&\hskip 2cm=\int_{(e_2,0)}^{(z_1,w_1)}\mathrm{d}r_1(z,w)+
\int_{(e_4,0)}^{(z_2,w_2)}\mathrm{d}r_1(z,w)+\int_{(e_6,0)}^{(z_3,w_3)}\mathrm{d}r_1(z,w)\\
&-\zeta_2(\boldsymbol{u})+2\mathfrak{n}_2+\frac12\frac{w_1}{(z_1-z_2)(z_1-z_3)}+\text{permutations} \\&\hskip 2cm =\int_{(e_2,0)}^{(z_1,w_1)}\mathrm{d}r_2(z,w)+
\int_{(e_4,0)}^{(z_2,w_2)}\mathrm{d}r_2(z,w)+\int_{(e_6,0)}^{(z_3,w_3)}\mathrm{d}r_2(z,w)\\
&-\zeta_3(\boldsymbol{u})+2\mathfrak{n}_3=\int_{(e_2,0)}^{(z_1,w_1)}\mathrm{d}r_3(z,w)+
\int_{(e_4,0)}^{(z_2,w_2)}\mathrm{d}r_3(z,w)+\int_{(e_6,0)}^{(z_3,w_3)}\mathrm{d}r_3(z,w) \ ,
\end{split} \label{zetaformula3}
\end{align}
where $\mathfrak{n}_j=\sum_{i=1}^{3}\eta^{\prime}_{ji}{\varepsilon}'_{i} + \eta_{ji}{\varepsilon}_{i}$. Here the characteristics ${\varepsilon}'_{i}$ and ${\varepsilon}_{i}$ of $\boldsymbol{K}_{\infty}$ are not reduced. Choosing $(z_1,w_1)=(Z,W),(z_2,w_2)=(e_4,0)$, $(z_3,w_3)=(e_6,0) $ we get from~\eqref{zetaformula3}
\begin{eqnarray}
-\zeta_1\left(\int_{(e_2,0)}^{(Z,W)} \mathrm{d}\boldsymbol{u} + \boldsymbol{K}_{\infty}\right)+2\mathfrak{n}_1+\frac12\frac{W(Z-e_4+e_6)}{(Z-e_4)(Z-e_6)}&=&\int_{(e_2,0)}^{(Z,W)}\mathrm{d}r_1(z,w) \ , \nonumber \\
-\zeta_2\left(\int_{(e_2,0)}^{(Z,W)} \mathrm{d}\boldsymbol{u} + \boldsymbol{K}_{\infty}\right)+2\mathfrak{n}_2+\frac12\frac{W}{(Z-e_4)(Z-e_6)}&=&\int_{(e_2,0)}^{(Z,W)}\mathrm{d}r_2(z,w) \ , \\
-\zeta_3\left(\int_{(e_2,0)}^{(Z,W)} \mathrm{d}\boldsymbol{u} + \boldsymbol{K}_{\infty} \right)+2\mathfrak{n}_3&=&\int_{(e_2,0)}^{(Z,W)}\mathrm{d}r_3(z,w) \ . \nonumber
 \label{zetaformula3_2}
\end{eqnarray}

The inversion formula for the integral of the third kind~\eqref{main1-2} is written as

\begin{align}\begin{split}
& W\int_{P'}^P\frac{1}{x-Z}\frac{\mathrm{d}x}{y}
=-2\left(\boldsymbol{u}^T-{\boldsymbol{u}'}^T\right)
\int_{(e_2,0)}^{(Z,W)} \mathrm{d}\boldsymbol{r}
+\mathrm{ln} \frac{\sigma\left(\boldsymbol{u}-\boldsymbol{v} - \boldsymbol{K}_{\infty} \right)}{\sigma\left(\boldsymbol{u}+\boldsymbol{v} - \boldsymbol{K}_{\infty}  \right)}-\mathrm{ln} \frac{\sigma\left(\boldsymbol{u}'-\boldsymbol{v} - \boldsymbol{K}_{\infty} \right)}{\sigma\left(\boldsymbol{u}'+\boldsymbol{v} - \boldsymbol{K}_{\infty}  \right)}
\end{split} \label{thirdinv3}
\end{align}
with
\[\boldsymbol{v}=\int_{(e_2,0)}^{(Z,W)}\mathrm{d}\boldsymbol{u},\quad \boldsymbol{u}=\int_{\infty}^{P}\mathrm{d}\boldsymbol{u},\quad \boldsymbol{u}'=\int_{\infty}^{P'}\mathrm{d}\boldsymbol{u} \]
and $\boldsymbol{u}\in\Theta_1$, $\boldsymbol{u}'\in\Theta_1$. The integrals $\displaystyle{\int_{(e_2,0)}^{(Z,W)} \mathrm{d}\boldsymbol{r}}$ are given by the formula~\eqref{zetaformula3_2}.

\section{Hyperelliptic curve of genus four}\label{sec:curvesgenus4}

As the next example we consider the hyperelliptic curve $X_4$ of genus four with nine real zeros given by
\begin{align}\begin{split}
w^2=4\prod_{k=1}^9(z-e_k)=4z^{9}+\lambda_8z^8+\ldots+\lambda_1z+\lambda_0 \,.
\end{split}
\label{curve4}
\end{align}
All calculations in this section will be done without explicit plotting of the homology basis.

The complete set of holomorphic and meromorphic differentials with a unique pole at infinity is
\begin{align}
\mathrm{d}u_1 & = \frac{\mathrm{d}z}{w}\,, & \qquad \mathrm{d}r_1 & = z(\lambda_3+2\lambda_4 z+3\lambda_5 z^2 + 4 \lambda_6 z^3 + 5 \lambda_7z^4+6\lambda_8 z^5 + 28 z^6)
 \frac{\mathrm{d}z}{4w} \,, \nonumber \\
\mathrm{d}u_2 &= \frac{z\mathrm{d}z}{w}\,, & \qquad \mathrm{d}r_2 & = z^2(\lambda_5+2\lambda_6z+3\lambda_7z^2+4\lambda_8z^3+20z^4)\frac{\mathrm{d}z}{4w} \,, \nonumber \\
\mathrm{d}u_3 &= \frac{z^2\mathrm{d}z}{w}\,, & \qquad \mathrm{d}r_2 & = z^3(\lambda_7+2\lambda_8z+12z^2)\frac{\mathrm{d}z}{4w}\,, \\
\mathrm{d}u_4 &= \frac{z^3\mathrm{d}z}{w}\,, & \qquad  \mathrm{d}r_3 & = \frac{z^4\mathrm{d}z}{w} \,. \nonumber
\end{align}


The Jacobi inversion problem for the equations
\begin{align}\begin{split}
\int_{\infty}^{(z_1,w_1)}\frac{ \mathrm{d}z}{w}
+\int_{\infty}^{(z_2,w_2)}\frac{ \mathrm{d}z}{w}
+\int_{\infty}^{(z_3,w_3)}\frac{ \mathrm{d}z}{w}
+\int_{\infty}^{(z_4,w_4)}\frac{ \mathrm{d}z}{w}
=u_1,\\
\int_{\infty}^{(z_1,w_1)}\frac{z \mathrm{d}z}{w}
+\int_{\infty}^{(z_2,w_2)}\frac{z \mathrm{d}z}{w}
+\int_{\infty}^{(z_3,w_3)}\frac{z \mathrm{d}z}{w}
+\int_{\infty}^{(z_4,w_4)}\frac{ z\mathrm{d}z}{w}
=u_2,\\
\int_{\infty}^{(z_1,w_1)}\frac{z^2 \mathrm{d}z}{w}
+\int_{\infty}^{(z_2,w_2)}\frac{z^2 \mathrm{d}z}{w}
+\int_{\infty}^{(z_3,w_3)}\frac{z^2 \mathrm{d}z}{w}
+\int_{\infty}^{(z_4,w_4)}\frac{z^2 \mathrm{d}z}{w}
=u_3,\\
\int_{\infty}^{(z_1,w_1)}\frac{z^3 \mathrm{d}z}{w}
+\int_{\infty}^{(z_2,w_2)}\frac{z^3 \mathrm{d}z}{w}
+\int_{\infty}^{(z_3,w_3)}\frac{z^3 \mathrm{d}z}{w}
+\int_{\infty}^{(z_4,w_4)}\frac{z^3 \mathrm{d}z}{w}
=u_4
\end{split} \label{JIP4}
\end{align}
is solved by
\begin{align}
\begin{split}
\sum^4_{i=1}z_i&=\wp_{44}(\boldsymbol{u}), \quad z_1z_2+z_1z_3+z_2z_3+z_1z_4+z_2z_4+z_3z_4=-\wp_{34}(\boldsymbol{u}), \nonumber \\
& z_1z_2z_3+z_4z_1z_2+z_4z_3z_1+z_4z_3z_2= \wp_{24}(\boldsymbol{u}), \quad z_1z_2z_3z_4=-\wp_{14}(\boldsymbol{u})\\
w_k&=\wp_{444}(\boldsymbol{u})z_k^3+\wp_{344}(\boldsymbol{u})z_k^2+\wp_{244}(\boldsymbol{u})z_k + \wp_{144}(\boldsymbol{u}) , \quad k=1,\ldots,4 \,.
\end{split} \label{SOLJIP4}
\end{align}

\subsection{Characteristics in genus four}

Let $\mathfrak{A}_k$ be the abelian image of the $k$-th branch point, namely
\begin{equation}
\boldsymbol{\mathfrak{A}}_k=\int_{\infty}^{(e_k,0)} \mathrm{d}\boldsymbol{u}= 2\omega \boldsymbol{\varepsilon}_k+2\omega' \boldsymbol{\varepsilon}_k', \quad k=1,\ldots,10 \,,
\end{equation}
where $\boldsymbol{\varepsilon}_k$ and $\boldsymbol{\varepsilon}_k'$ are column vectors whose entries $\varepsilon_{k,j}$, $\varepsilon'_{k,j}$ are $\frac{1}{2}$ or $0$ for all $k=1,\ldots,8$, $j=1,2,3$.

The characteristics of the branch points in the fixed homology basis yield
\begin{align}\begin{split}
[\boldsymbol{{\mathfrak A}}_1]= \frac{1}{2}
\begin{pmatrix} 1 & 0 & 0 & 0 \\ 0 & 0 & 0 & 0 \end{pmatrix}\, ,\quad
[\boldsymbol{{\mathfrak A}}_2]= \frac{1}{2}
\begin{pmatrix} 1 & 0 & 0 & 0 \\ 1 & 0 & 0 & 0 \end{pmatrix}\, ,\quad
[\boldsymbol{{\mathfrak A}}_3] = \frac{1}{2}
\begin{pmatrix} 0 & 1 & 0 & 0 \\ 1 & 0 & 0 & 0 \end{pmatrix}\, , \\
[\boldsymbol{{\mathfrak A}}_4]= \frac{1}{2}
\begin{pmatrix} 0 & 1 & 0 & 0 \\ 1 & 1 & 0& 0 \end{pmatrix}\, ,\quad
[\boldsymbol{{\mathfrak A}}_5] = \frac{1}{2}
\begin{pmatrix} 0 & 0 & 1& 0 \\ 1 & 1 & 0& 0 \end{pmatrix}\, ,\quad
[\boldsymbol{{\mathfrak A}}_6]= \frac{1}{2}
\begin{pmatrix} 0 & 0 & 1& 0 \\ 1 & 1 & 1& 0 \end{pmatrix}\, , \\
[\boldsymbol{{\mathfrak A}}_7]= \frac{1}{2}
\begin{pmatrix} 0 & 0 & 0 & 1 \\ 1 & 1 & 1 & 0 \end{pmatrix}\, ,\quad
[\boldsymbol{{\mathfrak A}}_8]= \frac{1}{2}
\begin{pmatrix} 0 & 0 & 0 & 1 \\ 1 & 1 & 1 & 1 \end{pmatrix}\, ,\quad
[\boldsymbol{{\mathfrak A}}_9]= \frac{1}{2}
\begin{pmatrix} 0 & 0 & 0 & 0 \\ 1 & 1 & 1 & 1 \end{pmatrix}\, , \\
[\boldsymbol{{\mathfrak A}}_{10}]= \frac{1}{2}
\begin{pmatrix} 0 & 0 & 0 & 0  \\ 0 & 0 & 0 & 0 \end{pmatrix}\, .
\end{split}
\label{hombasis_gen4}
\end{align}

The characteristics of the vector of Riemann constants $\boldsymbol{K}_{\infty}$ with the base point at infinity are even and singular as in the genus 3 example
\begin{equation}
[\boldsymbol{K}_{\infty}]=[\boldsymbol{\mathfrak A}_2]+[\boldsymbol{\mathfrak A}_4]+[\boldsymbol{\mathfrak A}_6]+[\boldsymbol{\mathfrak A}_8] = \frac12 \begin{pmatrix} 1 & 1 & 1 & 1 \\  0  & 1 & 0 & 1 \end{pmatrix}  \, .
\end{equation}

From the above characteristics 256 half-periods can be built as follows. If we start with singular even characteristics, then there should be only one such characteristic that corresponds to the vector of Riemann constants $\boldsymbol{K}_{\infty}$. The corresponding partition reads $\mathcal{I}_2\cup \mathcal{J}_2 = \{ \} \cup \{1, 2, \ldots, 10\}$ and the $\theta$-function $\theta(\boldsymbol{K}_{\infty}+\boldsymbol{v})$ vanishes at the origin $\boldsymbol{v}=0$ to the order $m=2$.

The half-periods $\boldsymbol{\Delta}_1=(2\omega)^{-1}\boldsymbol{\mathfrak A}_k+\boldsymbol{K}_{\infty}\in \widetilde{\Theta}_1$ correspond to partitions

\begin{equation}
\mathcal{I}_1\cup \mathcal{J}_1 =  \{ k,10 \} \cup \{ j_1,\ldots, j_8 \}, \quad j_1,\ldots,j_8 \notin \{10,k\}
\end{equation}
and the $\theta$-function $\theta(\boldsymbol{\Delta}_1+\boldsymbol{v})$ vanishes at the origin $\boldsymbol{v}=0$ to the order $m=2$ as follows from the Table~\ref{table1}. The characteristics of the 9 half-periods
\begin{equation}
\left[(2\omega)^{-1} \boldsymbol{\mathfrak A}_i+\boldsymbol{K}_{\infty}\right] =: \delta_i \, ,\quad i=1,\ldots,9
\end{equation}
are singular and even.

Also denote the $36$ half-periods that are images of two branch points
\begin{align}
\boldsymbol{\Omega}_{ij}=2\omega ( \boldsymbol{\varepsilon}_i+\boldsymbol{\varepsilon}_j)+2\omega' (\boldsymbol{\varepsilon}_i'+  \boldsymbol{\varepsilon}_j') ,\qquad  i,j = 1,\ldots, 9, i\neq j \ .
\end{align}
The half-periods $\boldsymbol{\Delta}_2=(2\omega)^{-1}\boldsymbol{\Omega}_{ij}
+\boldsymbol{K}_{\infty}\in \widetilde{\Theta}_2$ correspond to the partitions
\begin{equation}
\mathcal{I}_1\cup \mathcal{J}_1 = \{ i,j \} \cup \{ j_1,\ldots, j_8 \}, \quad j_1,\ldots,j_8 \notin \{i,j\}
\end{equation}
and the $\theta$-function $\theta(\boldsymbol{\Delta}_2 + \boldsymbol{v})$ vanishes at the origin, $\boldsymbol{v}=0$, to the order $m=1$.
Therefore the characteristics of the $36$ half-periods are nonsingular and odd
\begin{equation}
\left[(2\omega)^{-1} \boldsymbol{\Omega}_{ij}+\boldsymbol{K}_{\infty}\right] =: \varepsilon_{ij}   \, ,\quad 1\leq i<j\leq 9 \,.
\end{equation}
We introduce $84$ half-periods that are images of three branch points
\begin{align}
\boldsymbol{\Omega}_{ijk}=2\omega ( \boldsymbol{\varepsilon}_i+\boldsymbol{\varepsilon}_j+\boldsymbol{\varepsilon}_k)+ 2\omega' (\boldsymbol{\varepsilon}_i'+  \boldsymbol{\varepsilon}_j'+  \boldsymbol{\varepsilon}_k')\in \mathrm{Jac}(X_g) ,\quad  1\leq i<j<k\leq 9 \ ,
\end{align}
and $126$ half-periods that are images of four branch points
\begin{align}
\boldsymbol{\Omega}_{ijkl}=2\omega ( \boldsymbol{\varepsilon}_i+\boldsymbol{\varepsilon}_j+\boldsymbol{\varepsilon}_k+\boldsymbol{\varepsilon}_l)+ 2\omega' (\boldsymbol{\varepsilon}_i'+  \boldsymbol{\varepsilon}_j'+  \boldsymbol{\varepsilon}_k' \boldsymbol{\varepsilon}_l')\in \mathrm{Jac}(X_g) ,\quad  1\leq i<j<k\leq 9 \ .
\end{align}
The half-periods $\boldsymbol{\Delta}_3=(2\omega)^{-1}\boldsymbol{\Omega}_{ijk}+\boldsymbol{K}_{\infty}$ correspond to the partitions
\begin{equation}
\mathcal{I}_0\cup \mathcal{J}_0 =  \{ i,j,k,10 \} \cup \{ j_1,\ldots, j_6 \}, \quad j_1,\ldots,j_6 \notin \{i,j,k,10\}\,.
\end{equation}
The $\theta$-function $\theta(\boldsymbol{\Delta}_3+\boldsymbol{v})$ vanishes at the origin $\boldsymbol{v}=0$ to the order $m=1$. The $84$ characteristics
\begin{equation}
\varepsilon_{ijk}=  \left[ (2\omega)^{-1} \boldsymbol{\Omega}_{ijk}+\boldsymbol{K}_{\infty} \right],\quad 1\leq i<j<k\leq 9
\end{equation}
are odd and nonsingular as follows from the Table~\ref{table1}.
The half-periods $\boldsymbol{\Delta}_4=(2\omega)^{-1}\boldsymbol{\Omega}_{ijkl}+\boldsymbol{K}_{\infty}$ correspond to the partitions
\begin{equation}
\mathcal{I}_0\cup \mathcal{J}_0 =  \{ i,j,k,l,10 \} \cup \{ j_1,\ldots, j_5 \}, \quad j_1,\ldots,j_5 \notin \{i,j,k,l,10\}\,.
\end{equation}
The $\theta$-function $\theta(\boldsymbol{\Delta}_4+\boldsymbol{v})$ does not vanish at the origin $\boldsymbol{v}=0$. And the $126$ characteristics
\begin{equation}
\varepsilon_{ijkl}=  \left[ (2\omega)^{-1} \boldsymbol{\Omega}_{ijkl}+\boldsymbol{K}_{\infty} \right],\quad 1\leq i<j<k\leq 9
\end{equation}
are even and nonsingular.

Altogether there are $256=4^3$ characteristics classified by the partitions of the branch points.

\subsection{Inversion of a holomorphic integral}

For the case of genus four the formula~\eqref{matprev} reduces to
\begin{equation}
x=-\left.\frac{\sigma_{144}(\boldsymbol{u})}{\sigma_{244}(\boldsymbol{u})}\right|_{\sigma(\boldsymbol{u})=0, \sigma_4(\boldsymbol{u})=0, \sigma_{44}(\boldsymbol{u})=0,}, \qquad \boldsymbol{u}=(u_1,u_2,u_3,u_4)^T \, . \label{inversiongenus4}
\end{equation}

The four holomorphic integrals,
\begin{align}
\int_{\infty}^{(x,w)} \frac{\mathrm{d}z}{w}=u_1,\quad
\int_{\infty}^{(x,w)} \frac{z\mathrm{d}z}{w}=u_2,\quad
\int_{\infty}^{(x,w)} \frac{z^2\mathrm{d}z}{w}=u_3 \quad
\int_{\infty}^{(x,w)} \frac{z^3\mathrm{d}z}{w}=u_4
\end{align}
are inverted by the formula (\ref{inversiongenus4}).

From the solution of the Jacobi inversion problem~\eqref{JIP1} for any 4 roots $1\leq i<j<k \leq 9$ follows
\begin{eqnarray}
& e_i+e_j+e_k+e_l=\wp_{44}(\boldsymbol{\Omega}_{ijkl}), \quad e_ie_j+e_ie_k+e_je_k+e_ie_l+e_je_l+e_ke_l=-\wp_{34}(\boldsymbol{\Omega}_{ijkl}), \nonumber \\
& e_ie_je_k+e_le_ie_j+e_le_ke_i+e_le_ke_j= \wp_{24}(\boldsymbol{\Omega}_{ijkl}), \quad e_i e_j e_k e_l=-\wp_{14}(\boldsymbol{\Omega}_{ijkl}) \ .
 \label{SOLJIP4_Omegaijkl}
\end{eqnarray}
From equation~\eqref{relation} one finds the remaining components of the function $\wp_{ij}(\boldsymbol{u})$
\begin{eqnarray}
\label{JIPE4_wp}
&& \wp_{11}(\boldsymbol{\Omega}_{ijkl})=s_2S_5+s_4S_3 \, , \quad \wp_{12}(\boldsymbol{\Omega}_{ijkl})=-s_4S_2-s_1S_5 \, , \quad \wp_{13}(\boldsymbol{\Omega}_{ijkl})=S_5+s_4S_1 \,, \nonumber \\
&& \wp_{22}(\boldsymbol{\Omega}_{ijkl})=2S_5 +s_1S_4 + s_3 S_2 + 2 S_1 s_4 \, , \nonumber \\
&& \wp_{23}(\boldsymbol{\Omega}_{ijkl})=-s_3S_1 -S_4 - 2 s_4 \, ,
\wp_{33}(\boldsymbol{\Omega}_{ijkl})=S_3 + s_2 S_1 + 2 s_3 \, ,
\end{eqnarray}
where $s_l$ are the elementary symmetric functions of order $l$ of the branch points $e_i,e_j,e_k,e_l$ and $S_l$ are the elementary symmetric functions of order $l$ of the remaining branch points $\{1,\ldots,9\} \setminus \{i,j,k,l\}$.

From~\eqref{kappa} using~\eqref{SOLJIP4_Omegaijkl} and~\eqref{JIPE4_wp} one can find the expression for the matrix $\varkappa$. To do that consider the half-period $\boldsymbol{\Omega}_{1234}$
\begin{align}
\varkappa = -\frac12 \mathfrak{P}(\boldsymbol{\Omega}_{1234}) -\frac12 {(2\omega)^{-1}}^T H(\boldsymbol{\Omega}_{1234}) (2\omega)^{-1}
\end{align}
with
\begin{equation}
H(\boldsymbol{\Omega}_{1234})=\frac{1}{\theta[\varepsilon]}\left(\theta_{ij}[\varepsilon]\right)_{i,j=1,\ldots,4},\quad \varepsilon=\left( \begin{array}{cccc} \frac12&\frac12& \frac12&\frac12\\  \frac12&0&0& \frac12  \end{array}   \right) \ .
\end{equation}

For the branch points $e_1,\ldots, e_{10}$ for $g>3$ one can also use the formula~\cite{ehkkl11}
\begin{equation}
e_i = -\frac{\sigma_{3}}{\sigma_4}\left( \boldsymbol{\mathfrak A}_{\mathcal{I}'} \right)+\frac{\sigma_{2}}{\sigma_{3}}\left( \boldsymbol{\mathfrak A}_{\mathcal{I}''} \right) \ , \label{ei1}
\end{equation}
where $\mathcal{I}'=\mathcal{I}''\cup\{i\}$ with $\mathcal{I}''=\{i_1,\ldots, i_{2}\}$ and $i\neq 10$, $i\not \in \mathcal{I}''$.

The $\zeta$-formula~\eqref{JIPmeromorphic} for genus four reads

\begin{align}\begin{split}
&-\zeta_1(\boldsymbol{u})+2\mathfrak{n}_1+\frac12 \frac{w_1(z_1^2 - z_1( z_2+z_3+z_4) + z_2z_4 + z_3z_4 + z_2z_3 )}{(z_1-z_2)(z_1-z_3)(z_1-z_4)}+\text{permutations}\\&\hskip 2cm=\int_{(e_2,0)}^{(z_1,w_1)}\mathrm{d}r_1(z,w)+
\int_{(e_4,0)}^{(z_2,w_2)}\mathrm{d}r_1(z,w)+\int_{(e_6,0)}^{(z_3,w_3)}\mathrm{d}r_1(z,w)+\int_{(e_8,0)}^{(z_4,w_4)}\mathrm{d}r_1(z,w)\\
&-\zeta_2(\boldsymbol{u})+2\mathfrak{n}_2+\frac12 \frac{w_1(z_1-(z_2+z_3+z_4))}{(z_1-z_2)(z_1-z_3)(z_1-z_4)}+\text{permutations}\\&\hskip 2cm=\int_{(e_2,0)}^{(z_1,w_1)}\mathrm{d}r_2(z,w)+
\int_{(e_4,0)}^{(z_2,w_2)}\mathrm{d}r_2(z,w)+\int_{(e_6,0)}^{(z_3,w_3)}\mathrm{d}r_2(z,w)+\int_{(e_8,0)}^{(z_4,w_4)}\mathrm{d}r_2(z,w)\\
&-\zeta_3(\boldsymbol{u})+2\mathfrak{n}_3+\frac12\frac{w_1}{(z_1-z_2)(z_1-z_3)(z_1-z_4)}+\text{permutations} \\&\hskip 2cm =\int_{(e_2,0)}^{(z_1,w_1)}\mathrm{d}r_3(z,w)+
\int_{(e_4,0)}^{(z_2,w_2)}\mathrm{d}r_3(z,w)+\int_{(e_6,0)}^{(z_3,w_3)}\mathrm{d}r_3(z,w)+\int_{(e_8,0)}^{(z_4,w_4)}\mathrm{d}r_3(z,w)\\
&-\zeta_4(\boldsymbol{u})+2\mathfrak{n}_4=\int_{(e_2,0)}^{(z_1,w_1)}\mathrm{d}r_4(z,w)+
\int_{(e_4,0)}^{(z_2,w_2)}\mathrm{d}r_4(z,w)+\int_{(e_6,0)}^{(z_3,w_3)}\mathrm{d}r_4(z,w)+\int_{(e_8,0)}^{(z_4,w_4)}\mathrm{d}r_4(z,w) \ ,
\end{split}  \label{zetaformula_g4}
\end{align}
where $\mathfrak{n}_j=\sum_{i=1}^{4}\eta^{\prime}_{ji}{\varepsilon}'_{i} + \eta_{ji}{\varepsilon}_{i}$. Here the characteristics ${\varepsilon}'_{i}$ and ${\varepsilon}_{i}$ of $\boldsymbol{K}_{\infty}$ are not reduced.

\noindent
Choosing $(z_1,w_1)=(Z,W),(z_2,w_2)=(e_4,0)$, $(z_3,w_3)=(e_6,0)$, $(z_4,w_4)=(e_8,0) $ we get from (\ref{zetaformula_g4})

\begin{eqnarray}
-\zeta_1\left(\int_{(e_2,0)}^{(Z,W)} \mathrm{d}\boldsymbol{u} + \boldsymbol{K}_{\infty}\right)+2\mathfrak{n}_1+\frac12\frac{W(Z^2-Z(e_4+e_6+e_8) +e_4e_6 + e_4e_8+e_6e_8) }{(Z-e_4)(Z-e_6)(Z-e_8)}&=&\int_{(e_2,0)}^{(Z,W)}\mathrm{d}r_1(z,w) \ , \nonumber \\
-\zeta_2\left(\int_{(e_2,0)}^{(Z,W)} \mathrm{d}\boldsymbol{u} + \boldsymbol{K}_{\infty}\right)+2\mathfrak{n}_2+\frac12\frac{W(Z-(e_4+e_6+e_8))}{(Z-e_4)(Z-e_6)(Z-e_8)}&=&\int_{(e_2,0)}^{(Z,W)}\mathrm{d}r_2(z,w) \ , \nonumber \\
-\zeta_3\left(\int_{(e_2,0)}^{(Z,W)} \mathrm{d}\boldsymbol{u} + \boldsymbol{K}_{\infty}\right)+2\mathfrak{n}_3+\frac12\frac{W}{(Z-e_4)(Z-e_6)(Z-e_8)}&=&\int_{(e_2,0)}^{(Z,W)}\mathrm{d}r_3(z,w) \ , \\
-\zeta_4\left(\int_{(e_2,0)}^{(Z,W)} \mathrm{d}\boldsymbol{u} + \boldsymbol{K}_{\infty} \right)+2\mathfrak{n}_4&=&\int_{(e_2,0)}^{(Z,W)}\mathrm{d}r_4(z,w) \ . \nonumber
 \label{zetaformula_g4_2}
\end{eqnarray}

With~\eqref{zetaformula_g4_2} the inversion formula for the integral of the third kind~\eqref{main1-2} yields

\begin{align}\begin{split}
& W\int_{P'}^P\frac{1}{x-Z}\frac{\mathrm{d}x}{y}
=-2\left(\boldsymbol{u}^T-{\boldsymbol{u}'}^T\right)
\int_{(e_2,0)}^{(Z,W)} \mathrm{d}\boldsymbol{r}
+\mathrm{ln} \frac{\sigma\left(\boldsymbol{u}-\boldsymbol{v} - \boldsymbol{K}_{\infty} \right)}{\sigma\left(\boldsymbol{u}+\boldsymbol{v} - \boldsymbol{K}_{\infty}  \right)}-\mathrm{ln} \frac{\sigma\left(\boldsymbol{u}'-\boldsymbol{v} - \boldsymbol{K}_{\infty} \right)}{\sigma\left(\boldsymbol{u}'+\boldsymbol{v} - \boldsymbol{K}_{\infty}  \right)}
\end{split} \label{thirdinv4}
\end{align}
with
\[\boldsymbol{v}=\int_{(e_2,0)}^{(Z,W)}\mathrm{d}\boldsymbol{u},\quad \boldsymbol{u}=\int_{\infty}^{P}\mathrm{d}\boldsymbol{u},\quad \boldsymbol{u}'=\int_{\infty}^{P'}\mathrm{d}\boldsymbol{u} \]
and $\boldsymbol{u}\in\Theta_1$, $\boldsymbol{u}'\in\Theta_1$.

\section{Application: Solutions to the geodesic equation in Ho\v{r}ava-Lifshitz black hole space-times}~\label{sec:applHL}


Now we are applying our developed methods of integration of differentials of the first and third kind to the integration of the equations of motion of pointike test particles in Ho\v{r}ava--Lifshitz space--times. This class of space--times provides a quantum gravity space--time model which is power--counting renormalizeble and reduces to General Relativity in the infrared limit, i.e. at large distances. However, it faces the problem to violate Lorentz--symmetry at short distances. The main reason for that is that the model contains only higher order spatial derivatives in the action, while higher order temporal derivatives (which would lead to ghost degrees of freedom) do not appear \cite{horava1,horava2}.


\subsection{Equations of motion}

The metric of a spherically symmetric black hole in Ho\v{r}ava--Lifshitz gravity is given by 
\begin{equation}
ds^2 = N^2(r) dt^2 - f^{-1}(r) dr^2 - r^2 \left(d\theta^2 + \sin^2\theta d\varphi^2\right) \, .
\end{equation}
The Lagrangian for a point particle moving in this space--time reads
\begin{eqnarray}
\label{lagrangian_geo}
\mathcal{L}=g_{\mu\nu}\frac{dx^{\mu}}{ds}\frac{dx^{\nu}}{ds}=\varepsilon
=N^2\left(\frac{dt}{d\tau}\right)^{2}-\frac{1}{f}\left(\frac{dr}{d\tau}\right)^{2}-r^2
\left(\frac{d\theta}{d\tau}\right)^{2}-r^2 \sin^2\theta\left(\frac{d\varphi}{d\tau}\right)^{2}
 \ ,  \end{eqnarray}
where $\varepsilon=0$ for massless particles and $\varepsilon=1$ for massive particles, respectively.

The constants of motion are the energy $E$ and the angular momentum (direction and absolute value) of the particle. We choose $\theta=\pi/2$ to fix the direction of the angular momentum and have
\begin{eqnarray}
E := N^2 \frac{dt}{d\tau} \, , \qquad L_z := r^2 \frac{d\varphi}{d\tau}  \, .
\end{eqnarray}
Using these constants of motion we get
\begin{eqnarray}
\label{eq1}
&& \left(\frac{dr}{d\tau}\right)^2 = \frac{f}{N^2}\left(E^2 - V_{\rm eff}(r) \right)  \ , \\
&& \left(\frac{dr}{d\varphi}\right)^2 = \frac{r^4}{L_z^2} \frac{f}{N^2}\left(E^2 - V_{\rm eff}(r) \right)   \ ,
\end{eqnarray}
with the effective potential
\begin{equation}
\label{potential}
V_{\rm eff}(r) = N^2 \left(\varepsilon + \frac{L_z^2}{r^2}\right)  \, .
\end{equation}

Static spherically symmetric black hole solutions of this theory have been discussed in~\cite{ks,lu_mei_pope,park}. In all cases, the metric functions are of the form
\begin{equation}
N^2(r)=f(r)=1+c_1 r^2 - \sqrt{c_2 r^4 + c_3 r} \ , \label{HLcoefficients}
\end{equation}
where $c_1$, $c_2$ and $c_3$ are constants. Here, we will be interested in the case
\begin{equation}
c_1=-\Lambda_{\rm W} \ , \quad c_2=0 \ , \quad c_3=\alpha^2 \sqrt{-\Lambda_{\rm W}} \ , \label{constantsc1c2c3}
\end{equation}
where $\Lambda_{\rm W}$ is proportional to the negative cosmological constant and $\alpha \ge 4/3^{3/4}$ is an arbitrary parameter \cite{lu_mei_pope}.
The geodesic equations of point particles in the fields given by the solutions in~\cite{ks,park} cannot be treated analytically within the proposed scheme and are discussed elsewhere~\cite{ehkkls}.

The space-time metric with the choice of parameters~\eqref{constantsc1c2c3} considered here is possibly not astrophysically or cosmologically relevant due to the negative sign of the cosmological constant. But it could be interesting in the framework of the AdS/CFT correspondence~\cite{Maldacena, GuKlePo98}. Our motivation to study the motion of test particles in this space-time is more of mathematical character. Furthermore, our results concerning the particle motion in Ho\v{r}ava-Lifshitz space-times exhibit the same mathematical structure as for a number of space-times mentioned in the introduction. 

Using the substitution $q=\sqrt{r}$ we find that radial part of the geodesic equation is of the form
\begin{equation}
\label{geodesic}
 \left(\frac{1}{q}\frac{dq}{d\varphi}\right)^2 = P_k(q)  \ ,
\end{equation}
where $k=8$ with
\begin{equation}
 P_8(q)=\frac{1}{4L_z^2}\left(\varepsilon \Lambda_W q^8 +\varepsilon \alpha (-\Lambda_w)^{1/4} q^5 +(E^2 -\varepsilon+\Lambda_W L_z^2) q^4
+L_z^2 \alpha  (-\Lambda_W)^{1/4} q - L_z^2\right)
\end{equation}
for massive particles, while
$k=4$ for massless particles with
\begin{equation}
\label{4th}
P_4(q)=\frac{1}{4L_z^2}\left((E^2 +\Lambda_W L_z^2)q^4
+L_z^2 \alpha  (-\Lambda_W)^{1/4} q - L_z^2\right)  \ .
\end{equation}
We then find that
\begin{equation}
\label{phir_genk}
 \varphi-\varphi_0=\int_{q_0}^{q} \frac{dq}{q\sqrt{P_k(q)} }  \ .
\end{equation}

\subsection{Light rays}

For light we have $\varepsilon=0$. We write the 4th order polynomial~\eqref{4th} as $P_4(q)=b_4 q^4 + b_1 q + b_0$. Introducing a new variable $x$ with
\begin{equation}
q=\frac{1}{x} + q_4
\end{equation}
where $q_4$ is any root of $P_4(y)$ we find that \eqref{phir_genk} reduces to
\begin{equation}
\varphi-\varphi_0 = - \frac1q_4\int_{x_0}^{x}\frac{dx}{\sqrt{P_3(x)}} + \frac1q_4\int_{x_0}^{x} \frac{dx}{(1+q_4 x)\sqrt{P_3(x)}}
\end{equation}
with 
\begin{equation}
P_3(x)=(b_1+4 b_4 q_4^3)x^3 + 6b_4 q_4^2 x^2 + 4 b_4 q_4 x + b_4 =: a_3 x^3 + a_2x^2 + a_1 x + a_0 \, . 
\end{equation}
Using the substitutions
\begin{equation}
x=\gamma z + \beta \, , \qquad \gamma= \sqrt[3]{\frac{4}{a_3}} \, , \qquad \beta=-\frac{a_2}{3a_3}
\end{equation}
this can be brought to the Weierstra{\ss} form
\begin{equation}
\varphi-\varphi_0 = - \frac{\gamma}{q_4}\left[\int_{z_0}^{z}\frac{dz}{\sqrt{4z^3 -g_2 z-g_3}}  + \int_{z_0}^{z} \frac{dz}{(1+q_4 (\gamma z +\beta))\sqrt{4z^3-g_2 z + g_3}}  \right] \label{eqP3}
\end{equation}
with
\begin{equation}
 g_2=-\sqrt[3]{\frac{4}{a_3}}\left(\frac{3a_1a_3-a_2^2}{3a_3}\right) \ \ , \ \ g_3=-a_0 + \frac{a_1 a_2}{3a_3} - \frac{2a_2^3}{27a_3^2} \, .
\end{equation}

In order to invert the elliptic integrals we introduce $\nu$ such that
$\nu-\nu_0=\int^z_{z_0}\frac{dz^\prime}{\sqrt{4{z^\prime}^3 - g_2 z^\prime - g_3}}$. Then $z=\wp(v)$, where $v=\nu-\nu_0-\int^\infty_{z_0}\frac{dz}{\sqrt{4{z}^3 - g_2 z - g_3}}$. Using $\wp^\prime(v)=\sqrt{4\wp^3(v)-g_2\wp(v)-g_3}$ equation~\eqref{eqP3} can be rewritten in the form
\begin{eqnarray}
\varphi-\varphi_0 = -\frac{\gamma}{q_4}\left[
\int_{v_0}^{v} dv^\prime + \frac{1}{q_4\gamma} \int_{v_0}^{v} \frac{dv^\prime}{ \wp(v^\prime) - \wp(v_\wp) }  \right] \label{eqP3_2}  \ ,
\end{eqnarray}
which becomes
\begin{equation}
\varphi-\varphi_0  =  -\frac{\gamma}{q_4}\left[
v - v_0 + \frac{1}{q_4\gamma} \frac{1}{\wp^\prime(v_\wp)}
\Biggl( 2\zeta(v_\wp)(v-v_{0}) + \ln\frac{\sigma(v-v_\wp)}{\sigma(v+v_\wp)}
- \ln\frac{\sigma(v_{0}-v_\wp)}{\sigma(v_{0}+v_\wp)} \Biggr)  \right] \label{solk=4}   \ ,
\end{equation}
where $v_\wp$ is defined by $\wp(v_\wp)=-\frac{1+ q_4\beta}{q_4\gamma}$.
The  solution~\eqref{solk=4} gives $\varphi=\varphi(v)$ and the inversion yields $v=v(\varphi)$. We can then find $z(\varphi)$ and, thus, $r(\varphi)$, by substituting $v$ into $z=\wp(v)$.

\subsection{Motion of massive particles}

For massive particles we have $\varepsilon = 1$ and, thus, $k=8$. Introducing the new coordinate $z$ through
\begin{equation}
q = \frac{1}{z} + q_8 \ ,
\end{equation}
where $q_8$ is any root of $P_8(q)$ we find from~\eqref{phir_genk}
\begin{equation}
 \varphi-\varphi_0=-\frac1q_8\int_{z_0}^{z}\frac{z^2 dz}{\sqrt{P_7(z)}}  + \frac{1}{q_8^2}\int_{z_0}^{z} \frac{z dz}{\sqrt{P_7(z)}}
- \frac{1}{q_8^3}\int_{z_0}^{z}\frac{dz}{\sqrt{P_7(z)}} + \frac{1}{q_8^4}\int_{z_0}^{z} \frac{dz}{(q_8^{-1} + z)\sqrt{P_7(z)}}   \ . \label{phiz_P7}
\end{equation}

The curve $w^2 = P_7(z)$ is a hyperelliptic curve of genus $g=3$. We then introduce $v$ such that
$v-v_0=\int^z_{z_0}\frac{dz^\prime}{\sqrt{P_7(z^\prime)}}$. The solution of this integral is $z(v)=-\frac{\sigma_{13}(\boldsymbol{u})}{\sigma_{23}(\boldsymbol{u})}$ (see Eq.~\eqref{onishi}), where
\begin{equation}
\boldsymbol{u}=\boldsymbol{\mathfrak A}_i+\left( \begin{array}{c} v - v_0 \\  f_1(v - v_0) \\ f_{2}(v - v_0) \end{array}   \right) ,\quad f_1(0)=f_2(0)=0 \ , \label{u1u2u3}
\end{equation}
and where the functions $f_1(v - v_0)$ and $f_2(v - v_0)$ can be found from the conditions $\sigma(\boldsymbol{u})=0$ and $\sigma_3(\boldsymbol{u})=0$. Also $z_{0}$ is chosen as a branch point of the polynomial $\mathcal{P}_7(z)$ which defines the half--integer characteristic $\boldsymbol{\mathfrak A}_i$~\cite{ehkkl11}.

Through a comparison with the $u_i = \int du_i$ from \eqref{holomorphicdiff} we obtain from \eqref{phiz_P7}
\begin{equation}
\varphi - \varphi_0 = - \frac{1}{q_8} f_{2}(v - v_0) + \frac{1}{q_8^2}f_{1}(v - v_0) - \frac{1}{q_8^3}(v - v_0) + \frac{1}{q_8^4} \int_{z_0}^{z} \frac{dz}{(q_8^{-1} + z)\sqrt{P_7(z)}}   \, . \label{phiz_P7_2}
\end{equation}
Here the last differential in the equation above is of the third kind and was discussed in section~\ref{subsec:main1} for arbitrary genus of the underlying polynomial curve (see equation~\eqref{main1-2}). Here $Z=-q_8^{-1}$, $W=\sqrt{P_7(Z)}$.
Then the solution of~\eqref{phiz_P7} is
\begin{eqnarray}
\varphi - \varphi_0 & = & - \frac{1}{q_8}f_{2}(v - v_0) + \frac{1}{q_8^2}f_{1}(v - v_0)
- \frac{1}{q_8^3}(v - v_0) \nonumber \\
 & & + \frac{1}{q_8^4 W} \left[ 2 \left(\int_{z_{0}}^{z} d\boldsymbol{u}\right)^T \left( \boldsymbol{\zeta} \left( \int_{(e_2,0)}^{(Z,W)} \mathrm{d} \boldsymbol{u} + \boldsymbol{K}_\infty  \right) - 2( \boldsymbol{\eta}^{\prime}\boldsymbol{\varepsilon}^\prime + \boldsymbol{\eta}\boldsymbol{\varepsilon} )  - \frac12 \boldsymbol{\mathfrak{Z}}(Z,W)  \right) \right. \nonumber \\ 
& & \left. + \ln\frac{\sigma\left(  W_-(z)  \right)}{\sigma\left( W_+(z) \right)}
-  \ln \frac{\sigma\left(  W_-(z_0)  \right)}{\sigma\left( W_+(z_0) \right)}    \right]   \ . \label{phiz_P7_4}
\end{eqnarray}
where $W_{\pm}(z) = \int^{z}_{\infty}{d\boldsymbol{u}} \pm  \int_{(e_2,0)}^{(Z,W)} \mathrm{d} \boldsymbol{u} - \boldsymbol{K}_\infty $.

The solution~\eqref{phiz_P7_4} represents a generalization of the case of genus 1 given in~\eqref{solk=4}.


\section{Conclusion and outlook}\label{sec:conclusions}


In this paper we developed the inversion of general hyperelliptic integrals of the first, second and third kind. Besides that, in \eqref{JIPmeromorphic} we explicitly solved the integration of meromorphic differentials in terms of the $\boldsymbol{\zeta}$--function, the $\boldsymbol{\mathfrak{Z}}$--vector and the half--periods $\boldsymbol{\eta}$ and $\boldsymbol{\eta^\prime}$. Moreover, we provided a method which reduces the number of periods which need to be calculated explicitly. We pointed out that computer algebra should be used for the calculation of the period matrices in any arbitrary basis which provides a quick and convenient method for the calculation of the matrix $\varkappa$ and the meromorphic half--periods. For this one needs to know the components of the matrix $\wp_{ij}$ which can be easily calculated with the help of Lemma~\ref{lemma:p-matrix}.

As a first example we applied this method for solving the geodesic equation in particular cases of Ho\v{r}ava--Lifshitz black hole space--times. We considered special cases related to underlying algebraic curves of genus one and three and presented the associated analytical solutions for the geodesic equations of massless and massive test particles. Other examples where this method will be applied are geodesics in Myers--Perry space--times \cite{MyersPerry86} and in black ring space--times \cite{EmparanReall02,PomeranskySenkov06}. 

When trying to solve the geodesic equation for general Ho\v{r}ava--Lifshitz black hole space--times, that is, for $c_2 \neq 0$ in \eqref{HLcoefficients}, there is no way to get rid of the square root. In this case one has to square the whole equation, thus arriving at a differential equation which is based on a {\it quartic} algebraic curve. Similarly, quartic problems also appear, for instance, for the geodesic motion in string theory inspired space--times such as Gau{\ss}--Bonnet space--times \cite{BoulwareDeser85}, as well as for the motion of charged particles in the space--time of the regular black hole given by Ayon-Beato and Garcia \cite{AyonBeatoGarcia98,Garciaetal11}. 

\section*{Acknowledgements} The authors would like to thank Yu. Fedorov, P. Richter and E. Hackmann for fruitful discussions. V.K. and P.S. acknowledge the financial support of the German Research Foundation DFG, and V.E. acknowledges gratefully financial support from the  Hanse--Wissenschaftskolleg (Institute for Advanced Study) in Delmenhorst as well as its hospitality. C.L. thanks the center of excellence QUEST for support.

\providecommand{\bysame}{\leavevmode\hbox to3em{\hrulefill}\thinspace}

\end{document}